\RequirePackage{fix-cm}
\documentclass[smallextended]{svjour3}       
\smartqed  
\usepackage{graphicx}
 \journalname{}

\usepackage{graphicx}
\usepackage{amsmath}
\usepackage{amscd}
\usepackage{mathptmx}

\usepackage{mathrsfs}
\usepackage{amstext}
\usepackage{amssymb}
\usepackage{color}
\usepackage[symbol]{footmisc}

\begin{document}

\title{The Noether--Bessel-Hagen Symmetry Approach for Dynamical Systems }

\titlerunning{The Noether--Bessel-Hagen Symmetry Approach for Dynamical Systems}        

\author{Zbyn\v{e}k Urban\footnote{Corresponding author.} \and Francesco Bajardi \and Salvatore~Capozziello}

\authorrunning{Z. Urban, F. Bajardi, S. Capozziello} 

\institute{Z. Urban \at
              Department of Mathematics, Faculty of Civil Engineering\\
V\v{S}B-Technical University of Ostrava, Ludv\'{i}ka Pod\'{e}\v{s}t\v{e} 1875/17, 708 33 Ostrava-Poruba,
Czech Republic\\
\email{zbynek.urban@vsb.cz}\\ \\
F. Bajardi  \at
              Dipartimento di Fisica "Ettore Pancini", Universit\'a degli Studi di Napoli Federico II 
              Compl. Univ. Monte S. Angelo, 80126 Napoli, Italy\\
              INFN Sez. di Napoli, Compl. Univ. di Monte S. Angelo, Edificio G, Via	Cinthia, I-80126, Napoli, Italy \\
               \email{Bajardi@na.infn.it} \\ \\
        S. Capozziello \at
             Dipartimento di Fisica "Ettore Pancini", Universit\'a degli Studi di Napoli Federico II 
              Compl. Univ. Monte S. Angelo, 80126 Napoli, Italy\\
              INFN Sez. di Napoli, Compl. Univ. di Monte S. Angelo, Edificio G, Via	Cinthia, I-80126, Napoli, Italy \\
              Gran Sasso Science Institute, viale F. Crispi 7, I-67100, L'Aquila, Italy\\
              Tomsk State Pedagogical University, ul. Kievskaya, 60, 634061 Tomsk, Russia\\   
               \email{Capozziello@na.infn.it}        
           }
\date{Received: date / Accepted: date}

\maketitle

\begin{abstract}
The Noether-Bessel-Hagen theorem can be considered a natural extension of Noether Theorem to search for symmetries. Here, we develop the approach for dynamical systems introducing the basic foundations of the method. Specifically, we establish the Noether--Bessel-Hagen analysis of mechanical systems where external forces are present. In the second part of the paper, the approach is adopted to select symmetries for a given systems. In particular, we focus on the case of harmonic oscillator as a testbed for the theory, and on a cosmological system derived from scalar-tensor gravity with unknown scalar-field potential $V (\varphi)$. We show that the shape of potential is selected by the presence of symmetries. The approach results particularly useful as soon as the Lagrangian of a given system is not immediately identifiable or it is not a Lagrangian system.

\keywords{Lagrangian\and Noether symmetry approach\and
Noether--Bessel-Hagen symmetry\and invariant differential form\and
fibered mechanics\and  extended gravity cosmology}
 \subclass{70H33 \and 58E30 \and 58D19 \and 53A15}
\end{abstract}

\section{Introduction}

Noether symmetries have proved to be one of the most prolific mathematical tools for identifying conserved quantities and reducing dynamical systems. However, sometime it is difficult to identify the Lagrangian or the Hamiltonian related to a given system and then apply the so-called Noether Symmetry Approach
(cf. Capozziello \textit{et al.} \cite{Capozziello:1996bi,Dialektopoulos:2018qoe}).

Our aim, in this paper, is to analyze Noether-type symmetries of dynamical
systems with \emph{external forces}, which leave \emph{invariant}
the corresponding equations of motion, known as the \emph{Noether\textendash Bessel-Hagen
symmetries} (Bessel-Hagen \cite{Bessel-Hagen}); for symmetries in local variational theories see monographs by Kossmann-Schwarzbach \cite{Kossmann} and Krupka \cite{Krupka-Book}, and recent papers e.g. \cite{BK1,BK2,UB,Cattafi,PW,Francaviglia,Bashirov}. 
As we will show, this is a~straightforward generalization of the Noether Approach that can result extremely useful in several areas of physics like mechanics, field theory, cosmology and, in general, dynamical systems.

Let us recall
the well-known fact that any symmetry of a~Lagrangian $\lambda$
is also a~symmetry of the associated Euler\textendash Lagrange form
$E_{\lambda}$ or, in other words, any Noether symmetry (assumed by
the \emph{Noether's first theorem} \cite{Noether}) of a~Lagrangian $\lambda=\mathscr{L}dt$
is also a~Noether\textendash Bessel-Hagen symmetry of the mechanical
system derivable from $\lambda$. 

For \emph{variational} systems,
that is systems derivable as the Euler\textendash Lagrange equations
of a~Lagrangian, we have both the Noether currents for invariant
Lagrangians and an extension of Noether currents for invariant Euler\textendash Lagrange
forms; to this purpose, the infinitesimal first variation formula is
utilized. However, mechanical systems are very often \emph{not} variational,
typical examples of those include mechanical systems with frictional
forces. For non-variational mechanical systems, we may nevertheless introduce
the concept of invariance and study the Noether\textendash Bessel-Hagen
symmetries, although the standard Noether conserved currents are not
at disposal in this case. 

A~new idea of this work consists in solving the \emph{Noether\textendash Bessel-Hagen
equation} of Killing-type, 
\begin{equation}
\partial_{\xi}\varepsilon=0,\label{eq:NBH-intro}
\end{equation}
with respect to a  vector field $\xi$, an external force
$\phi$ and a~potential function $\mathscr{U}$, which are unknown. Here the mechanical
system is described by the source form $\varepsilon$ with components
\[
\varepsilon_{i}=E_{i}(\mathscr{L})-\phi_{i},
\]
given by the Euler\textendash Lagrange expressions $E_{i}(\mathscr{L})$
of a~mechanical Lagrange function $\mathscr{L}=\mathscr{T}-\mathscr{U}$
and force components $\phi_{i}$. In particular, this viewpoint enable
us to search for such mechanical system potentials arising from
symmetry requirements.

In Section 2, we summarize basic definitions and results on \emph{invariant}
differential forms, defined on jet prolongations of a~fibered manifold
over $1$-dimensional base (\emph{fibered mechanics}). As the subject
of this paper concerns mechanical systems, we may reasonably restrict
ourselves to \emph{second-order} equations of motion corresponding
to \emph{first-order} Lagrangians. For~general, higher-order \emph{field
theory} treatment on the subject and our main sources for the concept
of invariance, we refer to works by Trautman \cite{Trautman1,Trautman2} and Krupka \cite{Krupka-Book,DK-Invariant,Krupka-Invariance}, where all proofs can be found so we omit them in this section. For geometric invariance approach see also Sardanashvily
\cite{Sardanashvily}, and within classical Euclidean framework, Olver
\cite{Olver}, Bluman and Kumei \cite{Bluman} for symmetries of differential
equations, and Kossmann-Schwarzbach \cite{Kossmann} for Noether invariant
variational structures.

Section 3 contains results on symmetries of mechanical systems on with
external forces. The concept of a~force is introduced as a~$1$-form
on the tangent bundle $TM$ (cf. Krupka \cite{Krupka-Forces}). Our
main result, formulated in Theorem \ref{thm:Noether-BH}, describes
necessary and sufficient conditions for a~vector field $\xi$, external
force $1$-form $\phi$ and potential function $\mathscr{U}:M\rightarrow\mathbb{R}$
such that equation (\ref{eq:NBH-intro}) is satisfied identically.
These conditions include the Killing equation for $\xi$ with respect
to a~metric field $g$ defining the kinetic part of the Lagrangian,
and another condition on Lie derivaties of $\phi$ and $\mathscr{U}$;
note that the proof is based on the implication ``if metric $g$
is invariant, then the Levi-Civita connection $^{g}\Gamma$ associated
witth $g$ is also invariant''. Moreover, this theorem is further
developed for particular forms of external force $\phi$, including
conservative, dissipative, or variational forces. Closely related
works on the subject of symmetries and conservation laws for dissipative mechanical systems are Chien \emph{et. al.} \cite{Chien,Honein}, using a~different method based on variational multipliers
(``Neutral action method'').

In Section 4, the theory is applied  to  two examples: the $m$-dimensional
damped harmonic oscillator and a cosmological model derived from  scalar\textendash tensor theory of gravity. This latter example is particularly important for several reasons that will be discussed in Section \ref{scalar-tensor exp}. It is also worth noticing that the application of Noether theorem to different cosmological Lagrangians is an open issue aimed at selecting, by the existence of  symmetries, physically motivated models.  In fact, symmetries are able to reduce  dynamics and provide exact solutions of the field equations. For instance in \cite{Capozziello:2007wc,Paliathanasis:2014iva,Bajardi:2019zzs,Capozziello:2012iea}, the Noether Symmetry Approach is applied to some modified theories of gravity in a spherically symmetric background, while in \cite{Capozziello:2008ch,Capozziello:1996ay,Capozziello,Atazadeh:2011aa}, cosmological symmetries are studied. The approach applied in these papers (whose general features can be found in \cite{Capozziello:1996bi,Dialektopoulos:2018qoe}) needs a Lagrangian description to be performed; this is the main difference with respect to the approach  we are going to present in this paper, which can be directly applied to the field equations and does not require the system to be variational.

In this paper, the configuration space of a~mechanical system is
considered to be the second jet prolongation $J^{2}Y$ of a~product
fibered manifold $Y=\mathbb{R}\times M$, where $M$ is an open subset
of Euclidean space $\mathbb{R}^{m}$. This assumption allows us to
work with a~global chart. Jet spaces are then canonically identified
with products, namely $J^{1}Y\cong\mathbb{R}\times TM$, $J^{2}Y\cong\mathbb{R}\times T^{2}M$,
where $TM$ is the tangent space over $M$, endowed with global coordinates
$(x^{i},\dot{x}^{i})$, and $T^{2}M$ is the bundle of velocities
of order $2$ over $M$, endowed with global coordinates $(x^{i},\dot{x}^{i},\ddot{x}^{i})$.
Throughout, the Einstein summation convention is applied; $d\eta$,
$i_{\xi}\eta$, $\partial_{\xi}\eta$ denote the exterior derivative,
the contraction operation and the Lie derivative with respect to a~vector
field $\xi$ of a~differential form $\eta$. $\Gamma_{jk}^{i}$ are
the \emph{Christoffel symbols}, associated with a~metric tensor $g_{ij}$,
i.e.
\begin{equation}
\Gamma_{jk}^{i}=\frac{1}{2}g^{is}\left(\frac{\partial g_{sj}}{\partial x^{k}}+\frac{\partial g_{sk}}{\partial x^{j}}-\frac{\partial g_{jk}}{\partial x^{s}}\right).\label{eq:Christoffel}
\end{equation}
Along the paper, $\partial_{\xi}g$ and  $\partial_{\xi}\Gamma$ denote Lie derivatives
of metric field $g$ and of connection $\Gamma$ with respect to a~vector
field $\xi$,

\[
(\partial_{\xi}\Gamma)_{\eta}(\zeta)=(\partial_{\xi}\Gamma)_{jk}^{i}\zeta^{j}\eta^{k}\left(\frac{\partial}{\partial x^{i}}\right)_{x},
\]
where
\begin{equation}
(\partial_{\xi}\Gamma)_{jk}^{i}=\frac{\partial\Gamma_{jk}^{i}}{\partial x^{s}}\xi^{s}-\Gamma_{jk}^{s}\frac{\partial\xi^{i}}{\partial x^{s}}+\Gamma_{ks}^{i}\frac{\partial\xi^{s}}{\partial x^{j}}+\Gamma_{js}^{i}\frac{\partial\xi^{s}}{\partial x^{k}}+\frac{\partial^{2}\xi^{i}}{\partial x^{j}\partial x^{k}}.\label{eq:LieDerConnection}
\end{equation}

We use concepts and methods of global variational geometry, as described
in Section 2, and our results can be generalized to arbitrary smooth
fibered manifolds, including higher-order fibered mechanics and field
theory.

\section{Invariant source forms in fibered mechanics and the Noether\textendash Bessel-Hagen
equation}

Throughout, we consider a~fibered manifold $\pi:Y\rightarrow X$,
where $\dim X=1$, and its \emph{first} and \emph{second jet prolongations}
$J^{1}Y$ and $J^{2}Y$, respectively. The variational geometry structures,
well adapted to this work, can be found in \cite{Krupka-Book,KUV,UV,Volna}. For an open set $W\subset Y$, denote $W^{1}$
(respectively $W^{2}$) the preimage of $W$ in the canonical jet
projection $\pi^{1,0}:J^{1}Y\rightarrow Y$ (respectively $\pi^{2,0}:J^{2}Y\rightarrow Y$).
$\Omega^{1}W$ (respectively $\Omega^{2}W$) denotes the exterior
algebra of differential forms on $W^{1}$ (respectively $W^{2}$).
If $(V,\psi)$, $\psi=(t,x^{i})$, is a~fibered chart on $Y$, the
associated chart on $J^{1}Y$ (respectively $J^{2}Y$) reads $(V^{1},\psi^{1})$,
$\psi^{1}=(t,x^{i},\dot{x}^{i})$ (respectively $(V^{2},\psi^{2})$,
$\psi^{2}=(t,x^{i},\dot{x}^{i},\ddot{x}^{i})$). By means of chart,
we put $hdt=dt$, $hdx^{i}=\dot{x}^{i}dt$, $hd\dot{x}^{i}=\ddot{x}^{i}dt$,
and for any function $f:W^{1}\rightarrow\mathbb{R}$, $hf=f\circ\pi^{2,1}$,
where $\pi^{2,1}:J^{2}Y\rightarrow J^{1}Y$ is the canonical jet projection.
These formalae define a~\emph{global} homomorphism of exterior algebras
$h:\Omega^{1}W\rightarrow\Omega^{2}W$, called $\pi$\emph{-horizontalization}.
A~$1$-form $\rho\in\Omega_{1}^{1}W$ is called \emph{contact}, if
$h\rho=0$. In a~fibered chart, a~contact $1$-form $\rho$ is expressible
as a~linear combination $\rho=A_{i}\omega^{i}$ of contact $1$-forms
$\omega^{i}=dx^{i}-\dot{x}^{i}dt$. Any differential $1$-form $\rho\in\Omega_{1}^{1}W$
has a~unique decomposition $(\pi^{2,1})^{*}\rho=h\rho+p\rho$, where
$h\rho$, respectively $p\rho$, is the \emph{horizontal} \emph{component},
respectively \emph{contact component,} of $\rho$. Note that this
decomposition has a~generalization to arbitrary $k$-forms, see \cite{Krupka-Book}.
A vector field $\xi$ on $W\subset Y$ is said to be $\pi$\emph{-projectable},
if there exists a~vector field $\xi_{0}$ on $X$ such that $T\pi\circ\xi=\xi_{0}\circ\pi$.
In a~fibered chart, a~$\pi$-projectable vector field $\xi$ has
the expression $\xi=\xi^{0}(\partial/\partial t)+\xi^{i}(\partial/\partial x^{i})$,
where $\xi^{0}=\xi^{0}(t)$, $\xi^{i}=\xi^{i}(t,x^{j})$. For a~$\pi$-projectable
vector field $\xi$ , the vector field $J^{r}\xi$ on $W^{2}\subset J^{2}Y$
is defined by the formula
\[
J^{2}\xi(J^{2}\gamma)=\left.\frac{d}{dt}J^{2}\alpha_{t}^{\xi}(J^{2}\gamma)\right|_{t=0},
\]
where $\alpha_{t}^{\xi}$ is the local $1$-parameter group $\xi$,
and $J^{2}\alpha_{t}^{\xi}$ is the jet prolongation of automorphism
$\alpha_{t}^{\xi}$.

A~\emph{Lagrangian} of order $1$ for $Y$ is defined as a~$\pi^{1}$-horizontal
$1$-form $\lambda$ on $W^{1}\subset J^{1}Y$. In a~fibered chart
$(V,\psi)$, $\psi=(t,x^{i})$, on $W$, $\lambda$ has an expression
\begin{equation}
\lambda=\mathscr{L}dt,\label{eq:Lagrangian}
\end{equation}
where $\mathscr{L}:W^{1}\rightarrow\mathbb{R}$ is the (local) \emph{Lagrange
function}, associated with $\lambda$. The Euler\textendash Lagrange
mapping, well-known in the calculus of variations, assignes to a~Lagrangian
$\lambda$ the associated \emph{Euler\textendash Lagrange form} $E_{\lambda}$,
which is a~$1$-contact $2$-form on $W^{2}\subset J^{2}Y$, locally
expressed as
\begin{equation}
E_{\lambda}=E_{i}(\mathscr{L})\omega^{i}\wedge dt,\label{eq:EL-form}
\end{equation}
where its coefficients $E_{i}(\mathscr{L})$ are the \emph{Euler\textendash Lagrange
expressions},
\begin{equation}
E_{i}(\mathscr{L})=\frac{\partial\mathscr{L}}{\partial x^{i}}-\frac{d}{dt}\frac{\partial\mathscr{L}}{\partial\dot{x}^{i}}=\frac{\partial\mathscr{L}}{\partial x^{i}}-\frac{\partial^{2}\mathscr{L}}{\partial t\partial\dot{x}^{i}}-\frac{\partial^{2}\mathscr{L}}{\partial x^{j}\partial\dot{x}^{i}}\dot{x}^{j}-\frac{\partial^{2}\mathscr{L}}{\partial\dot{x}^{j}\partial\dot{x}^{i}}\ddot{x}^{j}.\label{eq:EL-expressions}
\end{equation}
The Euler\textendash Lagrange form (\ref{eq:EL-form}) is a~particular
example of a~\emph{source form} $\varepsilon$, which is by definition
a~$1$-contact, $\pi^{2,0}$-horizontal $2$-form on $W^{2}\subset J^{2}Y$.
We note that a~Lagrangian is a~representative of a~class of $1$-forms
and a~source form is a~representative of a~class of $2$-forms
in the variational sequence of quotient sheaves over $W\subset Y$;
see e.g. \cite{Volna}.

Recall now the well-known \emph{Cartan form} $\Theta_{\lambda}$ of
a~first-order Lagrangian $\lambda$. Suppose $\lambda\in\Omega_{1,X}^{1}W$
has a~chart expression (\ref{eq:Lagrangian}), then $\Theta_{\lambda}$
is locally expressed by
\begin{equation}
\Theta_{\lambda}=\mathscr{L}dt+\frac{\partial\mathscr{L}}{\partial\dot{x}^{i}}\omega^{i}.\label{eq:CartanForm}
\end{equation}
By means of chart transformations, one can easily observe that formula
(\ref{eq:CartanForm}) defines a~$1$-form $\Theta_{\lambda}\in\Omega_{1}^{1}W$,
which has the following properties:

(i) $h\Theta_{\lambda}=\lambda$, and

(ii) $i_{\xi}d\Theta_{\lambda}$ is contact $1$-form for every $\pi^{1,0}$-vertical
vector field $\xi$ on $W^{1}$.\\
We note that differential forms obeying properties (i) and (ii) are
called \emph{Lepage equivalents} (forms) of a~Lagrangian, see \cite{Krupka-Book,KKS}
and references therein. The meaning of condititon (i) is that variational
functionals defined by $\Theta_{\lambda}$ and $\lambda$ coincide.
Note also that for $\lambda\in\Omega_{1,X}^{1}W$, the Lepage equivalent
of $\lambda$ is by conditions (i) and (ii) determined in a~unique
way; this is \emph{no} longer true in higher-order field theory. Moreover,
the $1$-contact component of the exterior derivative of Cartan form
$\Theta_{\lambda}$ coincides with the Euler\textendash Lagrange form
$E_{\lambda}$,
\begin{equation}
p_{1}d\Theta_{\lambda}=E_{\lambda}.\label{eq:Theta-E}
\end{equation}

For a~$\pi$-projectable vector field $\xi$ on $W\subset Y$, and
for any section $\gamma$ of $\pi:Y\rightarrow X$ with values in
$W$, we have the \emph{infinitesimal first variation formula},
\begin{equation}
\begin{split}
J^{1}\gamma{}^{*}\partial_{J^{1}\xi}\lambda & =J^{2}\gamma{}^{*}i_{J^{2}\xi}E_{\lambda}+d\left(J^{1}\gamma{}^{*}i_{J^{1}\xi}\Theta_{\lambda}\right).\label{eq:FirstVar}
\end{split}
\end{equation}

A~diffeomorphism $\alpha:W\rightarrow Y$ is called an \emph{invariance
transformation} of $\lambda$, resp. $\varepsilon$, if 
\begin{equation}
(J^{r}\alpha){}^{*}\lambda=\lambda,\quad\textrm{resp.}\quad(J^{r}\alpha){}^{*}\varepsilon=\varepsilon,\label{eq:Invariance}
\end{equation}
where $J^{r}\alpha:W^{r}\rightarrow J^{r}Y$ is the $r$-jet prolongation
of $\alpha$. Note that this definition directly applies to vector
fields. A~$\pi$-projectable vector filed $\xi$ on $Y$ is called
a\emph{~generator of invariance transformations} of $\lambda$, resp.
$\varepsilon$, if its local one-parameter group $\alpha_{t}^{\xi}$
consists of invariance transformations of $\lambda$, resp. $\varepsilon$. 
\begin{lemma}
\label{lem:Invariance}Let $\lambda\in\Omega_{1,X}^{1}W$ be a~Lagrangian
of order $1$ for $Y$, and let $\varepsilon$ be a source form of
order $2$ for $Y$. A~$\pi$-projectable vector field $\xi$ on
$W\subset Y$ is a~generator of invariance transformations of $\lambda$,
respectively $\varepsilon$, if and only if the Lie derivative of
$\lambda$, respectively $\varepsilon$, with respect to $J^{r}\xi$
vanishes, i.e.
\begin{equation}
\partial_{J^{1}\xi}\lambda=0,\label{eq:NoetherEq}
\end{equation}
respectively
\begin{equation}
\partial_{J^{2}\xi}\varepsilon=0.\label{eq:NoetherBH-Eq}
\end{equation}

Generators of invariance transformations of $\lambda$, resp. $\varepsilon$,
form a~subalgebra of the algebra of vector fields on $W\subset Y$.
\end{lemma}

Equation (\ref{eq:NoetherEq}) is known as the \emph{Noether equation}
(cf. \cite{Trautman1}); (\ref{eq:NoetherBH-Eq}) is the geometric
formulation of the \emph{Noether\textendash Bessel-Hagen equation
}of the calculus of variations.

The classical (first) \emph{Noether's theorem,} which describes conservation
law equations for an extremal of an invariant Lagrangian,
is now a~straightforward consequence of formula (\ref{eq:FirstVar}).
\begin{theorem}
\label{thm:Noether}Let $\lambda\in\Omega_{1,X}^{1}W$ be a~Lagrangian
of order $1$ for $Y$, and let $\gamma$ be an extremal for $\lambda$.
Then for every generator $\xi$ of invariance transformations of $\lambda$,
\begin{equation}
d\left(J^{1}\gamma{}^{*}i_{J^{1}\xi}\Theta_{\lambda}\right)=0,\label{eq:NoetherConserved}
\end{equation}
where $\Theta_{\lambda}$ is the Cartan form (\ref{eq:CartanForm})
of $\lambda$.
\end{theorem}

\begin{remark}
The infinitesimal first variation formula implies also another consequence
for invariant Lagrangians. If $\xi$ is a~generator of invariance
transformations of $\lambda$, and a~section $\gamma$ of $Y$ satisfies
the conservation law equation (\ref{eq:NoetherConserved}), then the
Euler\textendash Lagrange expressions of $\lambda$ are \emph{linearly
dependent} along $\gamma$.
\end{remark}

Now, let source form $\varepsilon$ be (locally) variational, that
is $\varepsilon$ coincides with the Euler\textendash Lagrange form
$E_{\lambda}$ for some Lagrangian $\lambda\in\Omega_{1,X}^{r}W$.
For any diffeomorphism $\alpha:W\rightarrow Y$, the Euler\textendash Lagrange
form $E_{\lambda}$ obeys the formula
\begin{equation}
J^{2r}\alpha{}^{*}E_{\lambda}=E_{J^{r}\alpha{}^{*}\lambda},\label{eq:ELformInv}
\end{equation}
and for any vector field $\xi$ on $M$,
\begin{equation}
\partial_{J^{2}\xi}E_{\lambda}=E_{\partial_{J^{1}\xi}\lambda}\label{eq:ELformLie}
\end{equation}
(see \cite{Krupka-Book}). The next lemma is an immediate consequence
of formula (\ref{eq:ELformInv}).
\begin{lemma}
\emph{(i)} Every invariance transformation of $\lambda$ is an invariance
transformation of $E_{\lambda}$.

\emph{(ii)} If $\alpha$ is an invariance transformation of $E_{\lambda}$,
then $\lambda-J^{r}\alpha{}^{*}\lambda$ is a~variationally trivial
Lagrangian.
\end{lemma}

Invariant transformations of the Euler\textendash Lagrange form $E_{\lambda}$
extends the standard Theorem \ref{thm:Noether} of E. Noether. 
\begin{theorem}
\label{thm:NoetherExtended}Let $\lambda\in\Omega_{1,X}^{1}W$ be
a~Lagrangian of order $1$ for $Y$, let $\gamma$ be an extremal
for $\lambda$, and let $\xi$ be a~generator of invariance transformations
of the Euler\textendash Lagrange form $E_{\lambda}$. Then there exists
a~function $f$ on $W\subset Y$ such that
\begin{equation}
d\left(J^{1}\gamma{}^{*}(i_{J^{1}\xi}\Theta_{\lambda}+f)\right)=0.\label{eq:Noether-EL}
\end{equation}
\end{theorem}

\begin{remark}
Note that Theorem \ref{thm:NoetherExtended} contains conservation
law (\ref{eq:Noether-EL}) in a~\emph{global} form for the order
of Lagrangian equal $1$ only. For second and higher-order Lagrangians,
(\ref{eq:Noether-EL}) gives \emph{local} conservation laws.
\end{remark}

\section{The Noether\textendash Bessel-Hagen symmetries and external forces}

In this section, we study symmetries in sense of the Noether\textendash Bessel-Hagen
equation (\ref{eq:NoetherBH-Eq}) for source forms, representing mechanical
systems with external forces. 

Consider a~\emph{Lagrangian} $\lambda=\mathscr{L}dt$ for a~mechanical
system, where $\mathscr{L}:TM\rightarrow\mathbb{R}$ is a~first-order
\emph{Lagrange function} of the form kinetic \emph{minus} potential
energy,
\begin{equation}
\mathscr{L}=\mathscr{T}-\mathscr{U}.\label{eq:Lagrangian-T-U}
\end{equation}
The kinetic energy $\mathscr{T}$ is a~real-valued function defined
on $TM$ by the formula
\begin{equation}
\mathscr{T}=\frac{1}{2}g_{ij}\dot{x}^{i}\dot{x}^{j},\label{eq:Kinetic}
\end{equation}
where $g_{ij}$ is a~metric tensor on $M$. The potential energy
$\mathscr{U}$ is a~real-valued function defined on the configuration
space $M$, i.e. $\mathscr{U}=\mathscr{U}(x^{i})$, and describes
properties of the mechanical system. The \emph{Euler\textendash Lagrange
expressions} associated with $\mathscr{L}$ are functions on $T^{2}M$,
defined as
\begin{equation}
\begin{split}
E_{i}(\mathscr{L}) & =-\frac{\partial\mathscr{L}}{\partial x^{i}}+\frac{d}{dt}\frac{\partial\mathscr{L}}{\partial\dot{x}^{i}}.\label{eq:ELexp}
\end{split}
\end{equation}
Substituting the form (\ref{eq:Lagrangian-T-U}) of $\mathscr{L}$
into (\ref{eq:ELexp}), we get
\begin{equation}
\begin{split}
E_{i}(\mathscr{L}) & =E_{i}(\mathscr{T})+\frac{\partial\mathscr{U}}{\partial x^{i}}=-\frac{\partial\mathscr{T}}{\partial x^{i}}+\frac{d}{dt}\frac{\partial\mathscr{T}}{\partial\dot{x}^{i}}+\frac{\partial\mathscr{U}}{\partial x^{i}}\nonumber \\
 & =\frac{1}{2}\left(\frac{\partial g_{ij}}{\partial x^{k}}+\frac{\partial g_{ik}}{\partial x^{j}}-\frac{\partial g_{jk}}{\partial x^{i}}\right)\dot{x}^{j}\dot{x}^{k}+g_{ij}\ddot{x}^{j}+\frac{\partial\mathscr{U}}{\partial x^{i}}\label{eq:ELexp2}\\
 & =g_{ij}\left(\ddot{x}^{j}+\Gamma_{ks}^{j}\dot{x}^{k}\dot{x}^{s}\right)+\frac{\partial\mathscr{U}}{\partial x^{i}},\nonumber 
\end{split}
\end{equation}
where $\Gamma_{ks}^{j}$ are the \emph{Christoffel symbols} (\ref{eq:Christoffel})
associated to $g_{ij}$. Note that the Lagrange function $\mathscr{L}$
(\ref{eq:Lagrangian-T-U}) is regular since $g_{ij}$ is a~non-singular
matrix of the form
\[
g_{ij}=\frac{\partial^{2}\mathscr{L}}{\partial\dot{x}^{i}\partial\dot{x}^{j}}.
\]
Equations of motion of the mechanical system are the \emph{Euler\textendash Lagrange
equations associated with }$\mathscr{L}$\emph{, }
\begin{equation}
\begin{split}
 & E_{i}(\mathscr{L})=0,\quad i=1,\ldots,m.\label{eq:ELeq}
\end{split}
\end{equation}
The energy of the system equals
\begin{equation*}
\begin{split}
E=\frac{\partial\mathscr{L}}{\partial\dot{x}^{i}}\dot{x}^{i}-\mathscr{L} & =\frac{1}{2}g_{ij}\dot{x}^{i}\dot{x}^{j}+\mathscr{U},
\end{split}
\end{equation*}
hence the conservation law of energy along extremals reads
\begin{equation}
\frac{1}{2}g_{ij}\dot{x}^{i}\dot{x}^{j}+\mathscr{U}=0.\label{eq:EnergyCond}
\end{equation}

By an \emph{external force} for the mechanical system we call any~$\pi^{1,0}$-horizontal
$1$-form $\phi$ on $TM$, locally expressed as 
\begin{equation}
\phi=\phi_{i}dx^{i},\label{eq:Force}
\end{equation}
with components $\phi_{i}=\phi_{i}(x^{j},\dot{x}^{j})$ smooth real-valued
functitons on $TM$.

Consider equations of motion (\ref{eq:ELeq}) under the  influence of external
force $\phi=(\phi_{i})$,
\begin{equation}
\begin{split}
 & E_{i}(\mathscr{L})=\phi_{i},\quad i=1,\ldots,m.\label{eq:ELeq-F}
\end{split}
\end{equation}
 System (\ref{eq:ELeq-F}) associates a dynamical form $\varepsilon$
on $\mathbb{R}\times T^{2}M$,
\begin{equation}
\varepsilon=\varepsilon_{i}\omega^{i}\wedge dt,\label{eq:Source-F}
\end{equation}
where
\begin{equation}
\begin{split}
\varepsilon_{i} & =E_{i}(\mathscr{L})-\phi_{i},\label{eq:SourceCoef}
\end{split}
\end{equation}
and $\omega^{i}=dx^{i}-\dot{x}^{i}dt$ are contact $1$-forms on $\mathbb{R}\times TM$.
We call $\varepsilon$ (\ref{eq:Source-F}) the \emph{source form}
\emph{associated} with Lagrangian $\lambda=\mathscr{L}dt$ (\ref{eq:Lagrangian-T-U})
and force $\phi$ (\ref{eq:Force}).

Let $\xi$ be a vector field on $M$, locally expressed by
\begin{equation}
\xi=\xi^{i}\frac{\partial}{\partial x^{i}}.\label{eq:VectorField}
\end{equation}
Clearly, $\xi$ is $\pi$-projectable, and let $J^{2}\xi$ be its
second jet prolongation on $T^{2}M$,
\begin{equation}
J^{2}\xi=\xi^{i}\frac{\partial}{\partial x^{i}}+\dot{\xi}^{i}\frac{\partial}{\partial\dot{x}^{i}}+\ddot{\xi}^{i}\frac{\partial}{\partial\ddot{x}^{i}},\label{eq:2-Pro}
\end{equation}
where
\begin{equation}
\begin{split}
 & \dot{\xi}^{i}=\frac{d\xi^{i}}{dt}=\frac{\partial\xi^{i}}{\partial x^{j}}\dot{x}^{j},\quad\ddot{\xi}^{i}=\frac{d^{2}\xi^{i}}{dt^{2}}=\frac{\partial^{2}\xi^{i}}{\partial x^{j}\partial x^{k}}\dot{x}^{j}\dot{x}^{k}+\frac{\partial\xi^{i}}{\partial x^{j}}\ddot{x}^{j}.\label{eq:2-ProV}
\end{split}
\end{equation}

The Noether symmetries of $\lambda$ (\ref{eq:Lagrangian-T-U}) are
given by the following lemma.
\begin{lemma}
\emph{\label{lem:NoetherSym}}The following two conditions are equivalent:

\emph{(a)} A vector field $\xi$ on $M$ is a~generator of invariance
transformations of $\lambda$ (\ref{eq:Lagrangian-T-U}).

\emph{(b)} $\xi$ and $\mathscr{U}$ satisfy the Killing equations
\begin{equation}
\partial_{\xi}g=0,\quad\partial_{\xi}\mathscr{U}=0.\label{eq:NoetherSymCond}
\end{equation}
\end{lemma}

\begin{proof}
Equivalence of (a) and (b) is due to invariance Lemma \ref{lem:Invariance}.
Indeed, by a~straightforward calculation we get
\begin{equation*}
\begin{split}
\partial_{J^{1}\xi}\lambda & =\left(\partial_{J^{1}\xi}\mathscr{L}\right)dt=\left(\frac{\partial\mathscr{L}}{\partial x^{i}}\xi^{i}+\frac{\partial\mathscr{L}}{\partial\dot{x}^{i}}\dot{\xi}^{i}\right)dt\\
 & =\left(\frac{1}{2}\left(\frac{\partial g_{jk}}{\partial x^{i}}\xi^{i}+g_{ij}\frac{\partial\xi^{i}}{\partial x^{k}}+g_{ik}\frac{\partial\xi^{i}}{\partial x^{j}}\right)\dot{x}^{j}\dot{x}^{k}-\frac{\partial\mathscr{U}}{\partial x^{i}}\xi^{i}\right)dt\\
 & =\left(\frac{1}{2}(\partial_{\xi}g)\dot{x}^{j}\dot{x}^{k}-\partial_{\xi}\mathscr{U}\right)dt,
\end{split}
\end{equation*}
as required.
\end{proof}
Now, we formulate our main theorem, describing symmetries of source
form $\varepsilon$ (\ref{eq:Source-F}) associated with $\lambda$
and $\phi$. More precisely, our aim is to find conditions on a~vector
field $\xi$ on $M$, a~potential energy function $\mathscr{U}=\mathscr{U}(x^{i})$,
and force $\phi=(\phi_{i})$ such that the \emph{Noether\textendash Bessel-Hagen
equation},
\begin{equation}
\partial_{J^{2}\xi}\varepsilon=0,\label{eq:NBH-F}
\end{equation}
holds (see Lemma \ref{lem:Invariance}, (\ref{eq:NoetherBH-Eq})).
\begin{theorem}
\label{thm:Noether-BH}Let $\varepsilon$ be a source form on $\mathbb{R}\times T^{2}M$,
associated with $\lambda=\mathscr{L}dt$ (\ref{eq:Lagrangian-T-U})
and $\phi=(\phi_{i})$ (\ref{eq:Force}). The following two conditions
are equivalent:

\emph{(a)} A vector field $\xi$ on $M$ is a~generator of invariance
transformations of $\varepsilon$, i.e. the Noether\textendash Bessel-Hagen
equation (\ref{eq:NBH-F}) is satisfied identically.

\emph{(b)} $\xi$, $\mathscr{U}$, and $(\phi_{i})$ satisfy the following
the conditions
\begin{equation}
\partial_{\xi}g=0,\label{eq:NBH-F1}
\end{equation}
and, for every $i$,
\begin{equation}
\begin{split}
 & \frac{\partial}{\partial x^{i}}(\partial_{\xi}\mathscr{U})+(\partial_{J^{1}\xi}\phi)_{i}=0,\label{eq:NBH-F2}
\end{split}
\end{equation}
where
\begin{equation}
(\partial_{J^{1}\xi}\phi)_{i}=\frac{\partial\xi^{j}}{\partial x^{i}}\phi_{j}+\frac{\partial\phi_{i}}{\partial x^{j}}\xi^{j}+\frac{\partial\phi_{i}}{\partial\dot{x}^{j}}\frac{\partial\xi^{j}}{\partial x^{k}}\dot{x}^{k}\label{eq:ForceLieComp}
\end{equation}
represent the Lie derivative $\partial_{J^{1}\xi}\phi$ of external
force $1$-form $\phi=(\phi_{i})$ w.r.t. $J^{1}\xi$, and $\partial_{\xi}\mathscr{U}$
is the Lie derivative of function $\mathscr{U}$ w.r.t. $\xi$.
\end{theorem}

\begin{proof}
Equivalence of the conditions (a) and (b) follows from chart analysis
of equation (\ref{eq:NBH-F}). Indeed, computing the Lie derivative
of source form $\varepsilon$ (\ref{eq:Source-F}) with respect to
$J^{2}\xi$ (\ref{eq:2-Pro}), we obtain
\begin{equation*}
\begin{split}
\partial_{J^{2}\xi}\varepsilon & =\partial_{J^{2}\xi}(\varepsilon_{i}\omega^{i})\wedge dt=\left(\frac{\partial\varepsilon_{i}}{\partial x^{j}}\xi^{j}+\frac{\partial\varepsilon_{i}}{\partial\dot{x}^{j}}\dot{\xi}^{j}+\frac{\partial\varepsilon_{i}}{\partial\ddot{x}^{j}}\ddot{\xi}^{j}+\frac{\partial\xi^{j}}{\partial x^{i}}\varepsilon_{j}\right)dx^{i}\wedge dt\\
 & =\left(\frac{1}{2}\frac{\partial}{\partial x^{j}}\left(\frac{\partial g_{pq}}{\partial x^{i}}-\frac{\partial g_{ip}}{\partial x^{q}}-\frac{\partial g_{iq}}{\partial x^{p}}\right)\dot{x}^{p}\dot{x}^{q}-\frac{\partial g_{ip}}{\partial x^{j}}\ddot{x}^{p}-\frac{\partial^{2}\mathscr{U}}{\partial x^{i}\partial x^{j}}-\frac{\partial\phi_{i}}{\partial x^{j}}\right)\xi^{j}\\
 & +\left(\left(\frac{\partial g_{pj}}{\partial x^{i}}-\frac{\partial g_{ip}}{\partial x^{j}}-\frac{\partial g_{ij}}{\partial x^{p}}\right)\dot{x}^{p}-\frac{\partial\phi_{i}}{\partial\dot{x}^{j}}\right)\frac{\partial\xi^{j}}{\partial x^{k}}\dot{x}^{k}-g_{ij}\left(\frac{\partial^{2}\xi^{j}}{\partial x^{k}\partial x^{l}}\dot{x}^{k}\dot{x}^{l}+\frac{\partial\xi^{j}}{\partial x^{k}}\ddot{x}^{k}\right)\\
 & +\frac{\partial\xi^{j}}{\partial x^{i}}\left(\frac{1}{2}\left(\frac{\partial g_{pq}}{\partial x^{j}}-\frac{\partial g_{jp}}{\partial x^{q}}-\frac{\partial g_{jq}}{\partial x^{p}}\right)\dot{x}^{p}\dot{x}^{q}-g_{jp}\ddot{x}^{p}-\frac{\partial\mathscr{U}}{\partial x^{j}}-\phi_{j}\right)dx^{i}\wedge dt
\end{split}
\end{equation*}
\begin{equation*}
\begin{split}
 & =\left(\frac{1}{2}\frac{\partial}{\partial x^{j}}\left(\frac{\partial g_{pq}}{\partial x^{i}}-\frac{\partial g_{ip}}{\partial x^{q}}-\frac{\partial g_{iq}}{\partial x^{p}}\right)\xi^{j}\dot{x}^{p}\dot{x}^{q}\right.\\
 & +\frac{1}{2}\left(\left(\frac{\partial g_{qj}}{\partial x^{i}}-\frac{\partial g_{iq}}{\partial x^{j}}-\frac{\partial g_{ij}}{\partial x^{q}}\right)\frac{\partial\xi^{j}}{\partial x^{p}}+\left(\frac{\partial g_{pj}}{\partial x^{i}}-\frac{\partial g_{ip}}{\partial x^{j}}-\frac{\partial g_{ij}}{\partial x^{p}}\right)\frac{\partial\xi^{j}}{\partial x^{q}}\right)\dot{x}^{p}\dot{x}^{q}\\
 & +\frac{1}{2}\left(\frac{\partial g_{pq}}{\partial x^{j}}-\frac{\partial g_{jp}}{\partial x^{q}}-\frac{\partial g_{jq}}{\partial x^{p}}\right)\frac{\partial\xi^{j}}{\partial x^{i}}\dot{x}^{p}\dot{x}^{q}-g_{ij}\frac{\partial^{2}\xi^{j}}{\partial x^{p}\partial x^{q}}\dot{x}^{p}\dot{x}^{q}\\
 & -\frac{\partial\mathscr{U}}{\partial x^{j}}\frac{\partial\xi^{j}}{\partial x^{i}}-\frac{\partial^{2}\mathscr{U}}{\partial x^{i}\partial x^{j}}\xi^{j}-\phi_{j}\frac{\partial\xi^{j}}{\partial x^{i}}-\frac{\partial\phi_{i}}{\partial x^{j}}\xi^{j}-\frac{\partial\phi_{i}}{\partial\dot{x}^{j}}\frac{\partial\xi^{j}}{\partial x^{k}}\dot{x}^{k}\\
 & \left.-\left(\frac{\partial g_{ik}}{\partial x^{j}}\xi^{j}+g_{ij}\frac{\partial\xi^{j}}{\partial x^{k}}+g_{jk}\frac{\partial\xi^{j}}{\partial x^{i}}\right)\ddot{x}^{k}\right)dx^{i}\wedge dt.
\end{split}
\end{equation*}
Hence $\partial_{J^{2}\xi}\varepsilon$ vanishes identically if and
only if the following two conditions are satisfied:
\begin{equation}
\frac{\partial g_{ik}}{\partial x^{j}}\xi^{j}+g_{ij}\frac{\partial\xi^{j}}{\partial x^{k}}+g_{jk}\frac{\partial\xi^{j}}{\partial x^{i}}=0,\label{eq:Aux2}
\end{equation}
which is nothing but (\ref{eq:NBH-F1}), and
\begin{equation*}
\begin{split}
 & \frac{1}{2}\frac{\partial}{\partial x^{j}}\left(\frac{\partial g_{pq}}{\partial x^{i}}-\frac{\partial g_{ip}}{\partial x^{q}}-\frac{\partial g_{iq}}{\partial x^{p}}\right)\xi^{j}\dot{x}^{p}\dot{x}^{q}\\
 & +\frac{1}{2}\left(\left(\frac{\partial g_{qj}}{\partial x^{i}}-\frac{\partial g_{iq}}{\partial x^{j}}-\frac{\partial g_{ij}}{\partial x^{q}}\right)\frac{\partial\xi^{j}}{\partial x^{p}}+\left(\frac{\partial g_{pj}}{\partial x^{i}}-\frac{\partial g_{ip}}{\partial x^{j}}-\frac{\partial g_{ij}}{\partial x^{p}}\right)\frac{\partial\xi^{j}}{\partial x^{q}}\right)\dot{x}^{p}\dot{x}^{q}\\
 & +\frac{1}{2}g^{rj}\left(\frac{\partial g_{jp}}{\partial x^{q}}+\frac{\partial g_{jq}}{\partial x^{p}}-\frac{\partial g_{pq}}{\partial x^{j}}\right)\left(\frac{\partial g_{ir}}{\partial x^{s}}\xi^{s}+g_{is}\frac{\partial\xi^{s}}{\partial x^{r}}\right)\dot{x}^{p}\dot{x}^{q}-g_{ij}\frac{\partial^{2}\xi^{j}}{\partial x^{p}\partial x^{q}}\dot{x}^{p}\dot{x}^{q}\\
 & -\frac{\partial\mathscr{U}}{\partial x^{j}}\frac{\partial\xi^{j}}{\partial x^{i}}-\frac{\partial^{2}\mathscr{U}}{\partial x^{i}\partial x^{j}}\xi^{j}-\phi_{j}\frac{\partial\xi^{j}}{\partial x^{i}}-\frac{\partial\phi_{i}}{\partial x^{j}}\xi^{j}-\frac{\partial\phi_{i}}{\partial\dot{x}^{j}}\frac{\partial\xi^{j}}{\partial x^{k}}\dot{x}^{k}=0.
\end{split}
\end{equation*}
Applying the following standard identities,
\begin{equation*}
\begin{split}
 & \frac{\partial g_{ij}}{\partial x^{k}}=g_{is}\Gamma_{jk}^{s}+g_{js}\Gamma_{ik}^{s},\quad g_{ik}\Gamma_{pq}^{k}=\frac{1}{2}\left(\frac{\partial g_{ip}}{\partial x^{q}}+\frac{\partial g_{iq}}{\partial x^{p}}-\frac{\partial g_{pq}}{\partial x^{i}}\right),
\end{split}
\end{equation*}
the latter condition can be also rewritten as
\begin{equation}
\begin{split}
 & g_{ik}(\partial_{\xi}\Gamma)_{pq}^{k}\dot{x}^{p}\dot{x}^{q}+\frac{\partial}{\partial x^{i}}\left(\frac{\partial\mathscr{U}}{\partial x^{j}}\xi^{j}\right)+\phi_{j}\frac{\partial\xi^{j}}{\partial x^{i}}+\frac{\partial\phi_{i}}{\partial x^{j}}\xi^{j}+\frac{\partial\phi_{i}}{\partial\dot{x}^{j}}\frac{\partial\xi^{j}}{\partial x^{k}}\dot{x}^{k}=0,\label{eq:Aux1}
\end{split}
\end{equation}
where
\begin{equation}
(\partial_{\xi}\Gamma)_{pq}^{k}=\frac{\partial\Gamma_{pq}^{k}}{\partial x^{j}}\xi^{j}-\Gamma_{pq}^{j}\frac{\partial\xi^{k}}{\partial x^{j}}+\Gamma_{qj}^{k}\frac{\partial\xi^{j}}{\partial x^{p}}+\Gamma_{pj}^{k}\frac{\partial\xi^{j}}{\partial x^{q}}+\frac{\partial^{2}\xi^{k}}{\partial x^{p}\partial x^{q}}\label{eq:ConnecionLie}
\end{equation}
represent the Lie derivative $\partial_{\xi}\Gamma$ (\ref{eq:LieDerConnection})
of the Levi-Civita connection $\Gamma=(\Gamma_{pq}^{k})$ with respect
to $\xi$. However, since $^{g}\Gamma$ is invariant provided $g$
is invariant, it follows from (\ref{eq:Aux2}) that also (\ref{eq:ConnecionLie})
vanishes. Thus, (\ref{eq:Aux1}) already implies (\ref{eq:NBH-F2}).
This proves equivalence of the conditions (a) and (b).
\end{proof}
Reducing the form of external force $\phi=(\phi_{i})$, we obtain
important simplifications of Theorem \ref{thm:Noether-BH}. Consider
the following cases:

I. $\phi$ is identically zero, i.e. (\ref{eq:ELeq-F}) coincide with
the Euler\textendash Lagrange equations, associated with $\lambda=\mathscr{L}dt$.

II. $\phi$ is \emph{conservative}, i.e. its components $\phi_{i}$
are of the form
\begin{equation}
\phi_{i}=-\frac{\partial\mathscr{U}}{\partial x^{i}}.\label{eq:Conservative}
\end{equation}

III. $\phi$ is defined on $M$, i.e. its components $\phi_{i}$ are
functions of $x^{j}$ variables only.

IV. $\phi$ is \emph{dissipative}, i.e. its components $\phi_{i}$
are of the form

\begin{equation}
\phi_{i}=-\phi_{ij}\dot{x}^{j},\label{eq:DissipativeForce}
\end{equation}
where $\phi_{ij}:M\rightarrow\mathbb{R}$ is a~symmetric matrix,
and $\phi_{i}=-\partial F/\partial\dot{x}^{i}$ for a~\emph{dissipative
function},
\begin{equation}
F=\frac{1}{2}\phi_{ij}\dot{x}^{i}\dot{x}^{j},\label{eq:DissipativeFun}
\end{equation}
which is a~posite definite quadratic form in the dot variables (cf.
\cite{Landau}).

V. $\phi$ is \emph{variational}, i.e. $\phi_{i}$ coincide with the
Euler\textendash Lagrange expressions of a~Lagrange function, hence
\begin{equation}
\phi_{i}=\frac{\partial h}{\partial x^{i}}+\left(\frac{\partial\eta_{i}}{\partial x^{j}}-\frac{\partial\eta_{j}}{\partial x^{i}}\right)\dot{x}^{j},\label{eq:VariationalForce}
\end{equation}
for some functions $h=h(x^{j})$, $\eta_{i}=\eta_{i}(x^{j})$, on
$M$ (cf. \cite{Krupka-Forces}), which are free parameters. An alternative
approach using variational multipliers has been applied in \cite{Honein,Chien}.
\begin{remark}
For a~variational external force $\phi$ (Case V.), the source form
$\varepsilon$ (\ref{eq:Source-F}), associated to $\lambda=\mathscr{L}dt$
(\ref{eq:Lagrangian-T-U}) and $\phi$, is again variational and it
coincides with the Euler\textendash Lagrange form $E_{\tilde{\lambda}}$
of the Lagrangian $\tilde{\lambda}=\tilde{\mathscr{L}}dt$, where
\begin{equation}
\tilde{\mathscr{L}}=\mathscr{L}-h+\eta_{j}\dot{x}^{j}.\label{eq:AssocLagran}
\end{equation}
Note also that the external force $\tilde{\phi}=\tilde{\phi}_{i}dx^{i}$,
where 
\begin{equation}
\begin{split}
 & \tilde{\phi}_{i}=\phi_{i}-\frac{\partial h}{\partial x^{i}}=\left(\frac{\partial\eta_{i}}{\partial x^{j}}-\frac{\partial\eta_{j}}{\partial x^{i}}\right)\dot{x}^{j},\label{eq:AssocDissForce}
\end{split}
\end{equation}
is dissipative, see Case IV. for 
\[
\phi_{ij}=\frac{\partial\eta_{j}}{\partial x^{i}}-\frac{\partial\eta_{i}}{\partial x^{j}}.
\]
\end{remark}

\begin{corollary}
\label{cor:Case}Let $\varepsilon$ be a source form on $\mathbb{R}\times T^{2}M$,
associated with $\lambda=\mathscr{L}dt$ (\ref{eq:Lagrangian-T-U})
and $\phi=(\phi_{i})$, given by Cases I.\textendash V. A vector field
$\xi$ on $M$ is a~generator of invariance transformations of $\varepsilon$
if and only if $\xi$, $\mathscr{U}$, and $\phi=(\phi_{i})$ satisfy

\emph{(Case I.)}
\begin{equation}
\begin{split}
 & \partial_{\xi}g=0,\quad\partial_{\xi}\mathscr{U}=const.\label{eq:Case1}
\end{split}
\end{equation}

\emph{(Case II.)}
\begin{equation}
\partial_{\xi}g=0.\label{eq:Case2}
\end{equation}

\emph{(Case III.)}
\begin{equation}
\begin{split}
 & \partial_{\xi}g=0,\quad(\partial_{\xi}\phi)_{i}+\frac{\partial}{\partial x^{i}}(\partial_{\xi}\mathscr{U})=0.\label{eq:Case3}
\end{split}
\end{equation}

\emph{(Case IV.)}
\begin{equation}
\begin{split}
 & \partial_{\xi}g=0,\quad\partial_{J^{1}\xi}\phi=0,\quad\partial_{\xi}\mathscr{U}=const.\label{eq:Case4}
\end{split}
\end{equation}

\emph{(Case V.)}
\begin{equation}
\begin{split}
 & \partial_{\xi}g=0,\quad\partial_{J^{1}\xi}\tilde{\phi}=0,\quad\partial_{\xi}(\mathscr{U}+h)=const.,\label{eq:Case5}
\end{split}
\end{equation}
where $\tilde{\phi}$ is $1$-form given by (\ref{eq:AssocDissForce}).
\end{corollary}

\begin{proof}
Conditions (\ref{eq:Case1})\textendash (\ref{eq:Case5}) follow from
Theorem \ref{thm:Noether-BH}, where we substitute the form of external
force $\phi=(\phi_{i})$ with respect to Cases I\textendash V.
\end{proof}
The Noether theorem (Theorem \ref{thm:Noether}) and its extension
(Theorem \ref{thm:NoetherExtended}) have the following implications
for the mechanical Lagrangian (\ref{eq:Lagrangian-T-U}) and for its
completion by means of a~variatitonal external force.
\begin{corollary}
\label{cor:Noether-Add}Let $E_{\lambda}$ be the Euler\textendash Lagrange
form, associated to Lagrangian $\lambda=\mathscr{L}dt$ (\ref{eq:Lagrangian-T-U}).
Let $\xi$ be a vector field on $M$, and $\gamma$ be an extremal
for $\lambda$. Then

\emph{(i)} for every generator $\xi$ (\ref{eq:VectorField}) of invariance
transformations of $\lambda$,
\begin{equation}
dJ^{1}\gamma{}^{*}\left(g_{ij}\xi^{i}\dot{x}^{j}\right)=0,\label{eq:ConservedLambda}
\end{equation}
i.e. $\gamma$ is a geodesic with respect to the Levi-Civita connection
$^{g}\Gamma$.

\emph{(ii)} for every generator $\xi$ (\ref{eq:VectorField}) of
invariance transformations of $E_{\lambda}$,
\begin{equation}
\begin{split}
 & dJ^{1}\gamma{}^{*}\left(g_{ij}\xi^{i}\dot{x}^{j}+\mathscr{U}_{0}t\right)=0,\label{eq:ConservedE}
\end{split}
\end{equation}
where $\mathscr{U}_{0}\in\mathbb{R}$ is given by
\[
\mathscr{U}_{0}=\frac{\partial\mathscr{U}}{\partial x^{j}}\xi^{j}.
\]
\end{corollary}

\begin{proof}
1. Suppose $\xi$ is a~generator of invariance transformations of
$\lambda$. Substituting the form of Lagrangian $\lambda$ (\ref{eq:Lagrangian-T-U})
and the expression of the Cartan form,
\[
\Theta_{\lambda}=\left(\mathscr{L}-\frac{\partial\mathscr{L}}{\partial\dot{x}^{i}}\dot{x}^{i}\right)dt+\frac{\partial\mathscr{L}}{\partial\dot{x}^{i}}dx^{i},
\]
into Theorem \ref{thm:Noether}, we easily get the Noether current
given by (\ref{eq:ConservedLambda}). Moreover, using condition $\partial_{\xi}g=0$
(\ref{eq:NoetherSymCond}) from Lemma (\ref{lem:NoetherSym}), we
obtain
\begin{equation*}
\begin{split}
dJ^{1}\gamma{}^{*}\left(g_{ij}\xi^{i}\dot{x}^{j}\right) & =J^{1}\gamma{}^{*}d\left(g_{ij}\xi^{i}\dot{x}^{j}\right)\\
 & =J^{2}\gamma{}^{*}\left(\left(\frac{\partial g_{ij}}{\partial x^{k}}\xi^{i}+g_{ij}\frac{\partial\xi^{i}}{\partial x^{k}}\right)\dot{x}^{j}\dot{x}^{k}+g_{ik}\xi^{i}\ddot{x}^{k}\right)dt\\
 & =J^{2}\gamma{}^{*}\left(g_{ik}\xi^{i}\left(\ddot{x}^{k}+\Gamma_{js}^{k}\dot{x}^{j}\dot{x}^{s}\right)\right)dt.
\end{split}
\end{equation*}
Hence (\ref{eq:ConservedLambda}) holds if and only if $\gamma$ is
a~geodesic w.r.t. $(\Gamma_{jk}^{i})$.

2. Suppose $\xi$ is a~generator of invariance transformations of
$E_{\lambda}$. Since $\gamma$ is an extremal for $\lambda$, the
first variation formula (\ref{eq:FirstVar}) reduces to
\begin{equation}
J^{1}\gamma{}^{*}\partial_{J^{1}\xi}\lambda=d\left(J^{1}\gamma{}^{*}i_{J^{1}\xi}\Theta_{\lambda}\right).\label{eq:1stVarRed}
\end{equation}
From Corollary \ref{cor:Case} (Case I.), we have $\partial_{\xi}g=0$,
and
\begin{equation*}
\begin{split}
 & \frac{\partial\mathscr{U}}{\partial x^{j}}\xi^{j}=\mathscr{U}_{0}\in\mathbb{R},
\end{split}
\end{equation*}
hence the Lie derivative $\partial_{J^{1}\xi}\lambda$ of Lagrangian
(\ref{eq:Lagrangian-T-U}) reads
\begin{equation*}
\begin{split}
\partial_{J^{1}\xi}\lambda & =\left(\left(\frac{1}{2}\frac{\partial g_{jk}}{\partial x^{i}}\dot{x}^{j}\dot{x}^{k}-\frac{\partial\mathscr{U}}{\partial x^{i}}\right)\xi^{i}+g_{ij}\dot{x}^{j}\dot{\xi}^{i}\right)dt\\
 & =\left(\frac{1}{2}\left(\frac{\partial g_{jk}}{\partial x^{i}}\xi^{i}+g_{ij}\frac{\partial\xi^{i}}{\partial x^{k}}+g_{ik}\frac{\partial\xi^{i}}{\partial x^{j}}\right)\dot{x}^{j}\dot{x}^{k}-\frac{\partial\mathscr{U}}{\partial x^{i}}\xi^{i}\right)dt\\
 & =-\mathscr{U}_{0}dt.
\end{split}
\end{equation*}
Since $\partial_{J^{2}\xi}E_{\lambda}$ vanishes, we see at once that
$\partial_{J^{1}\xi}\lambda$ is a~variationally trivial Lagrangian
by means of formula (\ref{eq:ELformLie}). Applying this expression
into (\ref{eq:1stVarRed}), we get
\[
J^{2}\gamma{}^{*}\left(\mathscr{U}_{0}+g_{ik}\xi^{i}\left(\ddot{x}^{k}+\Gamma_{js}^{k}\dot{x}^{j}\dot{x}^{s}\right)\right)dt=0,
\]
or, equivalently, (\ref{eq:ConservedE}). This result completes Theorem
\ref{thm:NoetherExtended} for the mechanical Lagrangian $\lambda=(\mathscr{T}-\mathscr{U})dt$
(\ref{eq:Lagrangian-T-U}) whose Lie derivative $\partial_{J^{1}\xi}\lambda$
is of order zero.
\end{proof}
\begin{corollary}
Let $\varepsilon$ (\ref{eq:Source-F}) be the source form, associated
to Lagrangian $\lambda=\mathscr{L}dt$ (\ref{eq:Lagrangian-T-U})
and variational force $\phi=(\phi_{i})$ (\ref{eq:VariationalForce}).
Let $\xi$ be a vector field on $M$, and $\gamma$ be an extremal
for $\tilde{\lambda}$ (\ref{eq:AssocLagran}). Then

\emph{(i)} for every generator $\xi$ (\ref{eq:VectorField}) of invariance
transformations of $\tilde{\lambda}$,
\begin{equation}
dJ^{1}\gamma{}^{*}\left(\left(g_{ij}\dot{x}^{j}+\eta_{i}\right)\xi^{i}\right)=0.\label{eq:ConservedLambda-1}
\end{equation}

\emph{(ii)} for every generator $\xi$ (\ref{eq:VectorField}) of
invariance transformations of $E_{\tilde{\lambda}}$,
\begin{equation}
\begin{split}
 & dJ^{1}\gamma{}^{*}\left(\left(g_{ij}\dot{x}^{j}+\eta_{i}\right)\xi^{i}+\mathscr{U}_{0}t-f\right)=0,\label{eq:ConservedE-1}
\end{split}
\end{equation}
where $\mathscr{U}_{0}=\partial_{\xi}\mathscr{U}\in\mathbb{R}$ and
$f$ is a~solution of the first-order partial differential equation
\begin{equation}
\frac{\partial\eta_{j}}{\partial x^{i}}\xi^{i}+\eta_{i}\frac{\partial\xi^{i}}{\partial x^{j}}=\frac{\partial f}{\partial x^{j}},\label{eq:Aux3}
\end{equation}
which is integrable, and on a~star-shaped domain a solution of (\ref{eq:Aux3})
reads
\[
f=x^{j}\int_{0}^{1}\left(\frac{\partial\eta_{j}}{\partial x^{i}}\xi^{i}+\eta_{i}\frac{\partial\xi^{i}}{\partial x^{j}}\right)_{\left(sx^{k}\right)}ds.
\]
\end{corollary}

\begin{proof}
1. Analogously as in Corollary \ref{cor:Noether-Add}, formula (\ref{eq:ConservedLambda-1})
follows from the Noether theorem \ref{thm:Noether}.

2. Suppose $\xi$ is a~generator of invariance transformations of
$E_{\tilde{\lambda}}$ and $\gamma$ is an extremal for $\tilde{\lambda}$.
Applying Corollary \ref{cor:Case} (Case V.), namely the identities
$\partial_{\xi}g=0$, and $\partial_{\xi}(\mathscr{U}+h)=\mathscr{U}_{0}\in\mathbb{R}$,
we get
\begin{equation*}
\begin{split}
\partial_{J^{1}\xi}\tilde{\lambda} & =\left(\left(\frac{\partial\eta_{j}}{\partial x^{i}}\xi^{i}+\eta_{i}\frac{\partial\xi^{i}}{\partial x^{j}}\right)\dot{x}^{j}-\mathscr{U}_{0}\right)dt.
\end{split}
\end{equation*}
However, it easy to see that the identity $\partial\tilde{\phi}=0$
(\ref{eq:Case5}) is an integrability condition of the equation
\[
\left(\frac{\partial\eta_{j}}{\partial x^{i}}\xi^{i}+\eta_{i}\frac{\partial\xi^{i}}{\partial x^{j}}\right)\dot{x}^{j}=\frac{df}{dt}
\]
for unknown function $f=f(x^{k})$. Computing the contraction of the
Cartan form $i_{J^{1}\xi}\Theta_{\tilde{\lambda}}$, we get from the
first variation formula (\ref{eq:FirstVar}) the Noether conservation
law (\ref{eq:ConservedE-1}).
\end{proof}

\section{Examples}

\subsection{The damped harmonic oscillator}

Consider the equations of motion of the \emph{damped oscillator} in
$\mathbb{R}\times M$, where $M$ is an open subset of the Euclidean
space $\mathbb{R}^{m}$,
\begin{equation}
m_{ij}\ddot{x}^{j}+k_{ij}x^{j}=-\phi_{ij}\dot{x}^{j},\label{eq:DampedOscillatorEq}
\end{equation}
$i=1,\ldots,m$, where both the mass coefficients $m_{ij}$ and the
potential energy coefficients $k_{ij}$ are given by symmetric matrices over
the field of real numbers $\mathbb{R}$. Moreover, we assume that
$m_{ij}$ is a~non-singular; if this is not the case, we simply relax
the \emph{regularity condition} on the metric tensor $g$, cf. Section
3 (\ref{eq:Kinetic}). System (\ref{eq:DampedOscillatorEq}) belongs
to the class of mechanical systems with dissipative external forces,
cf. (\ref{eq:DissipativeForce}). The left-hand sides of (\ref{eq:DampedOscillatorEq})
are the Euler\textendash Lagrange expressions of the \emph{free oscillator}
Lagrangian
\begin{equation}
\mathscr{L}=\mathscr{T}-\mathscr{U}=\frac{1}{2}\left(m_{ij}\dot{x}^{i}\dot{x}^{j}-k_{ij}x^{i}x^{j}\right),\label{eq:LagrangianOsc}
\end{equation}
see (\ref{eq:Lagrangian-T-U}), and the \emph{energy} of the system
equals
\begin{equation}
\begin{split}
 & \frac{\partial\mathscr{L}}{\partial\dot{x}^{i}}\dot{x}^{i}-\mathscr{L}=\mathscr{T}+\mathscr{U}=\frac{1}{2}\left(m_{ij}\dot{x}^{i}\dot{x}^{j}+k_{ij}x^{i}x^{j}\right).\label{eq:EnergyOsc}
\end{split}
\end{equation}

The external force $\phi=(\phi_{i})$, see (\ref{eq:Force}), (\ref{eq:DissipativeForce}),
is given by 
\begin{equation}
\phi_{i}=-\phi_{ij}\dot{x}^{j},\label{eq:DissForceComp}
\end{equation}
also called the generalized frictional forces, see \cite{Landau}.
Note that system (\ref{eq:DampedOscillatorEq}) is \emph{not} variational
unless $\phi_{ij}=0$, which contradicts the posite definiteness of
the dissipative function (\ref{eq:DissipativeFun}). Source form $\varepsilon$,
associated with Lagrange function (\ref{eq:LagrangianOsc}) and dissipative
force $\phi$, is expressed by $\varepsilon=\varepsilon_{i}\omega^{i}\wedge dt$
(\ref{eq:Source-F}), where
\begin{equation}
\begin{split}
\varepsilon_{i} & =m_{ij}\ddot{x}^{j}+k_{ij}x^{j}+\phi_{ij}\dot{x}^{j}.\label{eq:SourceOsc}
\end{split}
\end{equation}

Let $\xi$ be a vector field on $M$, locally expressed by (\ref{eq:VectorField}),
and let $J^{2}\xi$ (\ref{eq:2-Pro}), (\ref{eq:2-ProV}), be its
second jet prolongation. Applying Corollary \ref{cor:Case} of Theorem
\ref{thm:Noether-BH}, we obtain conditions on symmetry $\xi$ on
$M$, potential coefficients $k_{ij}$, and dissipative force coefficients
$\phi_{ij}$ such that the Noether\textendash Bessel-Hagen equation
(\ref{eq:NBH-F}) for source form $\varepsilon$ holds.
\begin{theorem}
\label{cor:OSC}Let $\varepsilon$ (\ref{eq:SourceOsc}) be a source
form on $\mathbb{R}\times T^{2}M$, associated to the free oscillator
Lagrangian $\lambda=\mathscr{L}dt$ (\ref{eq:LagrangianOsc}) and
a dissipative force $\phi=(\phi_{1},\phi_{2})$ (\ref{eq:DissForceComp}).
Let $\xi$ be a~vector field on $M$, expressed by (\ref{eq:VectorField}).

\emph{(a)} $\xi$ is a~generator of invariance transformations of
$\lambda$ if and only if
\begin{equation}
\begin{split}
 & m_{ip}\frac{\partial\xi^{i}}{\partial x^{q}}+m_{iq}\frac{\partial\xi^{i}}{\partial x^{p}}=0,\label{eq:LambdaInvOscillator-Eq1}
\end{split}
\end{equation}
and
\begin{equation}
k_{ij}x^{j}\xi^{i}=0.\label{eq:LambdaInvOscillator-Eq2}
\end{equation}
The Noether conserved current (Corollary \ref{cor:Noether-Add}, (\ref{eq:ConservedLambda}))
along extremals for the free harmonic oscillator (i.e. $\phi=0$)
reads
\begin{equation}
\begin{split}
L & =m_{ij}\xi^{i}\dot{x}^{j}.\label{eq:ConservedCurrent-FreeOsc}
\end{split}
\end{equation}

\emph{(b)} $\xi$ is a~generator of invariance transformations of
$\varepsilon$ if and only if $\xi^{i}$ and $\phi_{ij}$ satisfy
(\ref{eq:LambdaInvOscillator-Eq1}) and
\begin{equation}
\begin{split}
 & k_{ij}\xi^{j}+k_{js}x^{s}\frac{\partial\xi^{j}}{\partial x^{i}}=0,\label{eq:SymEpsOsc-Eq2}
\end{split}
\end{equation}
\begin{equation}
\begin{split}
 & \phi_{is}\frac{\partial\xi^{s}}{\partial x^{j}}+\phi_{js}\frac{\partial\xi^{s}}{\partial x^{i}}+\frac{\partial\phi_{ij}}{\partial x^{s}}\xi^{s}=0,\label{eq:SymEpsOsc-Eq3}
\end{split}
\end{equation}
For $\phi=0$, along extremals for the free harmonic oscillator the
Noether\textendash Bessel-Hagen conserved current (Corollary \ref{cor:Noether-Add},
(\ref{eq:ConservedE})) reads
\begin{equation*}
\begin{split}
L= & m_{ij}\xi^{i}\dot{x}^{j}+k_{ij}x^{i}\xi^{j}t.
\end{split}
\end{equation*}
\end{theorem}

\begin{proof}
Assertions (a) and (b) are reformulations of Lemma \ref{lem:NoetherSym}
and Theorem \ref{thm:Noether-BH}, respectively. Indeed, conditions
(\ref{eq:LambdaInvOscillator-Eq1}), correspond to $\partial_{\xi}g=0$,
(\ref{eq:LambdaInvOscillator-Eq2}), (\ref{eq:SymEpsOsc-Eq2}) to
$\partial_{\xi}\mathscr{U}=0$, and (\ref{eq:SymEpsOsc-Eq3}) corresponds
to $\partial_{J^{1}\xi}\phi=0$, cf. Corollary \ref{cor:Case}, (\ref{eq:Case4}).
\end{proof}
\begin{remark}
In the case of the Euclidean metric $m_{ij}=k_{ij}=\delta_{ij}$,
conditions (\ref{eq:LambdaInvOscillator-Eq1}), (\ref{eq:LambdaInvOscillator-Eq2})
possess a~unique solution of the form $\xi^{i}=P_{j}^{i}x^{j}$,
where $P_{j}^{i}$ are some real numbers such that $P_{j}^{i}=-P_{i}^{j}$.
The Noether current (\ref{eq:ConservedCurrent-FreeOsc}) is $L=\delta_{ij}\xi^{i}\dot{x}^{j}$.
If $m=3$, for the $3$-dimensional free oscillator $L$ is expressed
as
\begin{equation*}
\begin{split}
L= & \sum_{j=1}^{3}\xi^{j}\dot{x}^{j}\\
= & -P_{2}^{1}r^{2}\dot{\varphi}\sin^{2}\vartheta+P_{3}^{1}r^{2}\left(\dot{\vartheta}\cos\varphi-\dot{\varphi}\sin\vartheta\cos\vartheta\sin\varphi\right)\\
 & +P_{3}^{2}r^{2}\left(\dot{\vartheta}\sin\varphi+\dot{\varphi}\sin\vartheta\cos\vartheta\cos\varphi\right),
\end{split}
\end{equation*}
which is a~linear combination of components of the \emph{angular
momentum} in the spherical coordinates $(r,\varphi,\vartheta)$ on
$M$.

Conditions (\ref{eq:LambdaInvOscillator-Eq1}), (\ref{eq:SymEpsOsc-Eq2}),
(\ref{eq:SymEpsOsc-Eq3}), with respect to the Euclidean metric, possess
a~unique solution of the form $\xi^{i}=P_{j}^{i}x^{j}$, where $P_{j}^{i}\in\mathbb{R}$
satisfy $P_{j}^{i}=-P_{i}^{j}$, and
\[
\phi_{is}P_{j}^{s}+\phi_{js}P_{i}^{s}+\frac{\partial\phi_{ij}}{\partial x^{s}}P_{l}^{s}x^{l}=0.
\]
For $\phi=0,$ the Noether\textendash Bessel-Hagen conserved current
(\ref{eq:ConservedE}) reads
\begin{equation*}
\begin{split}
L= & \sum_{j}\xi^{j}\dot{x}^{j}+\sum_{j}\xi^{j}x^{j}t
\end{split}
\end{equation*}
and it coincides with the angular momentum $L$ since $\xi^{j}x^{j}$
vanishes.
\end{remark}

\subsection{Scalar-tensor cosmological models in two-dimensional configuration space} \label{scalar-tensor exp}
Extended and modified theories of gravity \cite{Cai:2015emx,Capozziello:2011et} acquired a lot of interest in the recent years  due to some inconsistencies provided by   Einstein General Relativity which presents some shortcomings at ultraviolet and infrared scales. As an example, a grand unification theory in which gravity and the other fundamental interactions are included is missing so far, if one considers General Relativity as the final theory  describing the gravitational force. More fundamentally,  finding a link between General Relativity and Quantum Mechanics is one of the main goal of recent physics, and, despite of many attempts, Einstein's theory seems to be not the best candidate to this purpose. According to this point of view, many other approaches have been developed   as \emph{String Theory} \cite{Green:1987sp,Green:1987mn,Polchinski:1998rq,Polchinski:1998rr,Becker:2007zj}, \emph{Kaluza-Klein Theory} \cite{Clifton:2011jh,Han:1998sg}, \emph{Loop Quantum Gravity} \cite{Rovelli:2014ssa,Ashtekar:2011ni,Rovelli:1997yv}, \emph{Horava-Lifshitz Gravity} \cite{Kiritsis:2009sh,Cai:2009pe,Sotiriou:2010wn,Mukohyama:2010xz}, \emph{Non-Local Gravity} \cite{ArkaniHamed:2002fu,Modesto:2013jea,Modesto:2013ioa} \emph{etc.}. While General-Relativity turns out to be non-renormalizable from the two-loop level, the above mentioned theories are both renormalizable and unitary. 

However, ultra-violet scale issues do not represent the only problems of General Relativity: at the low-scale regime there are several incongruities between theory and observations emerging at astrophysical and cosmological scales. Some of them are  the anomalous accelerated expansion of the universe, the flat rotation curves of galaxies, the dynamics of cluster of galaxies  \emph{etc.}. These shortcomings can be overcame by assuming huge amount of dark energy and dark matter or relaxing the  strict hypothesis that  General Relativity is the theory describing gravitational interaction at any scale. Proposals are   modifying the Hilbert-Einstein action \cite{Sotiriou:2008rp,DeFelice:2010aj},  including the torsion \cite{Cai:2015emx,Hammond:2002rm,Arcos:2005ec}, and studying the related dynamics  \cite{Capozziello:2019klx,Capozziello:2011et,Capozziello:2007ec,Clifton:2011jh}.

Another issue of General Relativity is related  the early cosmology; the  accepted paradigm is that of \emph{inflation} which involves a scalar field  (or more than one scalar field), the \emph{inflaton} $\varphi$, to  generate  the primordial  accelerated phase capable of addressing the shortcomings of Cosmological Standard Model  \cite{Guth:1980zm,Guth:1982ec,Linde:1981mu}. In these models,   gravity is assumed  minimally or non-minimally  coupled to  $\varphi$. Theories in which the action contains both the Ricci scalar, the kinetic term $\dot{\phi}^2$, the potential $V(\varphi)$ and a  coupling $F(\varphi)$ are called \emph{scalar-tensor} theories of gravity \cite{Copeland:2006wr,Bezrukov:2007ep}.
A fundamental question is to select theories with reliable $F{\varphi}$ and $V(\varphi)$ capable of giving realistic cosmological models to be compared with observational data.  The Noether Symmetry Approach  proved extremely useful in selecting physically motivated model. See  \cite{Capozziello:1996bi} for a review. 

In this section, we apply the preceding  Noether\textendash Bessel-Hagen
Approach to the dynamical equations arising in  a particular  scalar-tensor cosmological model. As said above, finding  the form of the potential from symmetries is  important to derive   conserved quantities,  to reduce the dynamics and finally to find exact solutions that are always physically motivated \cite{Capozziello:1996bi}. Some examples can be found in \cite{Capozziello:1993vr,Capozziello:1994du,Paliathanasis:2014rja,Capozziello:1993yy,Borowiec:2014wva}.

Let us adopt the above method for   scalar-tensor cosmological models where external forces appear into dynamic with the aim to select by symmetries the scalar-field   potential and possible interacting terms. See for example \cite{ester}.

Let $M$ be an open subset of the Euclidean plane $\mathbb{R}^{2}$. The only coordinates are the cosmological scale factor $a$ and the scalar field $\varphi$, which form a~global chart on $M\subset\mathbb{R}^{2}$.
Consider the\textcolor{red}{{} }Klein-Gordon Lagrange function on $TM$,
given by
\begin{equation*}
\mathscr{L}(a,\varphi,\dot{a},\dot{\varphi})=3a\dot{a}^{2}-a^{3}\left(\frac{1}{2}\dot{\varphi}^{2}-V(\varphi)\right),
\end{equation*}
where $V$ is a~smooth function depending on $\varphi$ only as standard in cosmology \cite{Guth:1980zm,Guth:1982ec,Linde:1981mu,Copeland:2006wr,Bezrukov:2007ep,Capozziello:1993vr,Capozziello:1994du,Paliathanasis:2014rja,Capozziello:1993yy,Borowiec:2014wva}.
This Lagrange function is of the form (\ref{eq:Lagrangian-T-U}),
\begin{equation}
\mathscr{L}=\frac{1}{2}g_{11}\dot{a}^{2}+g_{12}\dot{a}\dot{\varphi}+\frac{1}{2}g_{22}\dot{\varphi}^{2}-\mathscr{U},\label{eq:Lagrange-KG}
\end{equation}
where
\begin{equation*}
g_{11}=6a,\quad g_{12}=0,\quad g_{22}=-a^{3},\quad\mathscr{U}(a,\varphi)=-a^{3}V(\varphi).
\end{equation*}

The associated Euler--Lagrange expressions are
\begin{eqnarray}
-\frac{\partial\mathscr{L}}{\partial a}+\frac{d}{dt}\frac{\partial\mathscr{L}}{\partial\dot{a}} &= -3\dot{a}^{2}+3a^{2}\left(\frac{1}{2}\dot{\varphi}^{2}-V(\varphi)\right)+\frac{d}{dt}\left(6a\dot{a}\right)\label{eq:EL1EGC}\\
&= 6a\ddot{a}+3\dot{a}^{2}+3a^{2}\left(\frac{1}{2}\dot{\varphi}^{2}-V(\varphi)\right),\nonumber
\end{eqnarray}
\begin{eqnarray}
 -\frac{\partial\mathscr{L}}{\partial\varphi}+\frac{d}{dt}\frac{\partial\mathscr{L}}{\partial\dot{\varphi}} &= \displaystyle -a^{3}\frac{dV}{d\varphi}+\frac{d}{dt}\left(-a^{3}\dot{\varphi}\right)=-a^{3}\frac{dV}{d\varphi}-3a^{2}\dot{a}\dot{\varphi}-a^{3}\ddot{\varphi},\label{eq:EL2EGC}
\end{eqnarray}
and the {\emph{energy condition} is
\begin{equation}
\label{eq:EnergyEGC}
\begin{split}
\frac{\partial\mathscr{L}}{\partial\dot{a}}\dot{a}+\frac{\partial\mathscr{L}}{\partial\dot{\varphi}}\dot{\varphi}-\mathscr{L} & =6a\dot{a}^{2}-a^{3}\dot{\varphi}^{2}-3a\dot{a}^{2}+a^{3}\left(\frac{1}{2}\dot{\varphi}^{2}-V(\varphi)\right)\\
 & =3a\dot{a}^{2}-a^{3}\left(\frac{1}{2}\dot{\varphi}^{2}+V(\varphi)\right).
\end{split}
\end{equation}
Hence,  the Euler--Lagrange equations read
\begin{eqnarray}
6a\ddot{a}+3\dot{a}^{2}+\frac{3}{2}a^{2}\dot{\varphi}^{2}-3a^{2}V(\varphi) & =0,\label{eq:ELeq1-EGC}\\
-a^{3}\ddot{\varphi}-3a^{2}\dot{a}\dot{\varphi}-a^{3}\frac{dV}{d\varphi} & =0.\label{eq:ELeq2-EGC}
\end{eqnarray}
Note that this system is equivalent to
\begin{eqnarray}
2\frac{\ddot{a}}{a}+\left(\frac{\dot{a}}{a}\right)^{2}+\frac{1}{2}\dot{\varphi}^{2}-V(\varphi) & =0,\label{eq:Ein1}\\
\ddot{\varphi}+3\frac{\dot{a}}{a}\dot{\varphi}+\frac{dV}{d\varphi} & =0,\label{eq:Ein2}
\end{eqnarray}
which is the system of equations in \cite{Capozziello:1994du}, with $F(\varphi) = 1/2$.
However, (\ref{eq:Ein1}), (\ref{eq:Ein2}) is \emph{not} variational in the sense discussed above.

Consider the equations of motion defined by (\ref{eq:EL1EGC}), (\ref{eq:EL2EGC}),
under influence of an external force $\left(f_{a},f_{\varphi}\right)$,
\begin{eqnarray}
6a\ddot{a}+3\dot{a}^{2}+\frac{3}{2}a^{2}\dot{\varphi}^{2}-3a^{2}V(\varphi) & =\phi_{a}\label{eq:EL1-EGC-Force}\\
-a^{3}\ddot{\varphi}-3a^{2}\dot{a}\dot{\varphi}-a^{3}\frac{dV}{d\varphi} & =\phi_{\varphi},\label{eq:EL2-EGC-Force}
\end{eqnarray}
where $\phi_{a}=\phi_{a}(a,\varphi,\dot{a},\dot{\varphi})$, $\phi_{\varphi}=\phi_{\varphi}(a,\varphi,\dot{a},\dot{\varphi})$
are smooth real-valued functions on $TM$.
The source form $\varepsilon$, associated with Lagrangian $\lambda=\mathscr{L}dt$
(\ref{eq:Lagrange-KG}) and force $\phi=\phi_{a}da+\phi_{\varphi}d\varphi$
(cf. (\ref{eq:Force})) is
\begin{equation}
\varepsilon=\left(\varepsilon_{a}\omega^{a}+\varepsilon_{\varphi}\omega^{\varphi}\right)\wedge dt,\label{eq:Source-EGC-Force}
\end{equation}
where its coefficients $\varepsilon_{a}$, $\varepsilon_{\varphi}$ are
real-valued functions on $T^{2}M$, given by
\begin{equation*}
\begin{split}
\varepsilon_{a} & =6a\ddot{a}+3\dot{a}^{2}+\frac{3}{2}a^{2}\dot{\varphi}^{2}-3a^{2}V(\varphi)-\phi_{a},\\
\varepsilon_{\varphi} & =-a^{3}\ddot{\varphi}-3a^{2}\dot{a}\dot{\varphi}-a^{3}\frac{dV}{d\varphi}-\phi_{\varphi},
\end{split}
\end{equation*}
and $\omega^{a}=da-\dot{a}dt$, $\omega^{\varphi}=d\varphi-\dot{\varphi}dt$
are contact $1$-forms on $\mathbb{R}\times TM$.

Let $\xi$ be a vector field on $M$, locally expressed by
\begin{equation}
\xi=\alpha(a,\varphi)\frac{\partial}{\partial a}+\beta(a,\varphi)\frac{\partial}{\partial\varphi},\label{eq:VectorField-2dim2}
\end{equation}
and let $J^{2}\xi$ be its second jet prolongation (see (\ref{eq:2-ProV})),
\begin{equation}
J^{2}\xi=\alpha\frac{\partial}{\partial a}+\beta\frac{\partial}{\partial\varphi}+\dot{\alpha}\frac{\partial}{\partial\dot{a}}+\dot{\beta}\frac{\partial}{\partial\dot{\varphi}}+\ddot{\alpha}\frac{\partial}{\partial\ddot{a}}+\ddot{\beta}\frac{\partial}{\partial\ddot{\varphi}},\label{eq:2ndProlong-2dim-2}
\end{equation}
where
\begin{equation*}
\begin{split}
 & \dot{\alpha}=\frac{\partial\alpha}{\partial a}\dot{a}+\frac{\partial\alpha}{\partial\varphi}\dot{\varphi},\quad\dot{\beta}=\frac{\partial\beta}{\partial a}\dot{a}+\frac{\partial\beta}{\partial\varphi}\dot{\varphi},\\
 & \ddot{\alpha}=\frac{\partial^{2}\alpha}{\partial a^{2}}\dot{a}^{2}+2\frac{\partial^{2}\alpha}{\partial a\partial\varphi}\dot{a}\dot{\varphi}+\frac{\partial^{2}\alpha}{\partial\varphi^{2}}\dot{\varphi}^{2}+\frac{\partial\alpha}{\partial a}\ddot{a}+\frac{\partial\alpha}{\partial\varphi}\ddot{\varphi},\\
 & \ddot{\beta}=\frac{\partial^{2}\beta}{\partial a^{2}}\dot{a}^{2}+2\frac{\partial^{2}\beta}{\partial a\partial\varphi}\dot{a}\dot{\varphi}+\frac{\partial^{2}\beta}{\partial\varphi^{2}}\dot{\varphi}^{2}+\frac{\partial\beta}{\partial a}\ddot{a}+\frac{\partial\beta}{\partial\varphi}\ddot{\varphi}.
\end{split}
\end{equation*}
The following theorem is a~reformulation of Theorem \ref{thm:Noether-BH}
for source form (\ref{eq:Source-EGC-Force}), associated with the
Klein\textendash Gordon Lagrangian and an external force; in this
case, a part of necessary and sufficient conditions expressed by the
Killing equation $\partial_{\xi}g=0$ (\ref{eq:NBH-F1}) can be integrated,
hence the symmetry $\xi$ is found explicitly.
\begin{theorem}
\label{thm:EGC}Let $\varepsilon$ (\ref{eq:Source-EGC-Force}) be
a source form on $\mathbb{R}\times T^{2}M$, associated to the Klein-Gordon
Lagrangian $\lambda$ and an external force $(\phi_{a},\phi_{\varphi})$
on $TM$. Let $\xi$ be a~vector field on $M$. The following two
conditions are equivalent:

\emph{(a)} A vector field $\xi$ on $M$ is a~generator of invariance
transformations of $\varepsilon$.

\emph{(b)} $\xi$ has an expression (\ref{eq:VectorField-2dim2}),
where
\begin{eqnarray}
& \alpha=A\frac{1}{\sqrt{a}}\left(\exp\left(\sqrt{\frac{3}{8}}\varphi\right)+C\exp\left(-\sqrt{\frac{3}{8}}\varphi\right)\right)\label{eq:NBH-EGCsym1}\\
& \beta=-\sqrt{6}A\frac{1}{a\sqrt{a}}\left(\exp\left(\sqrt{\frac{3}{8}}\varphi\right)-C\exp\left(-\sqrt{\frac{3}{8}}\varphi\right)\right)+B.\label{eq:NBH-EGCsym2}
\end{eqnarray}
Here $A,B\in\mathbb{R}$, and $C>0$ are real numbers, and $\varepsilon$
has an expression (\ref{eq:Source-EGC-Force}), where potential function
$V$ and forces $\phi_{a}$, $\phi_{\varphi}$ satisfy the equations
\begin{equation}
\begin{split}
 & \frac{9}{2}\alpha aV+3\beta a^{2}\frac{dV}{d\varphi}+6\frac{\partial\alpha}{\partial\varphi}a\frac{dV}{d\varphi}-\frac{\partial}{\partial a}\left(\alpha f_{a}\right)-\frac{\partial\beta}{\partial a}f_{\varphi}-\beta\frac{\partial f_{a}}{\partial\varphi}\label{eq:NBH-EGC-1}\\
 & \quad-\frac{\partial\phi_{a}}{\partial\dot{a}}\left(\frac{\partial\alpha}{\partial a}\dot{a}+\frac{\partial\alpha}{\partial\varphi}\dot{\varphi}\right)-\frac{\partial\phi_{a}}{\partial\dot{\varphi}}\left(\frac{\partial\beta}{\partial a}\dot{a}+\frac{\partial\beta}{\partial\varphi}\dot{\varphi}\right)=0,\nonumber 
\end{split}
\end{equation}
and 
\begin{equation}
\begin{split}
 & \frac{3}{2}\alpha a^{2}\frac{dV}{d\varphi}+\beta a^{3}\frac{d^{2}V}{d\varphi^{2}}+3\frac{\partial\alpha}{\partial\varphi}a^{2}V-\frac{\partial}{\partial\varphi}\left(\beta f_{\varphi}\right)-\frac{\partial\alpha}{\partial\varphi}f_{a}-\alpha\frac{\partial f_{\varphi}}{\partial a}\label{eq:NBH-EGC-2}\\
 & -\frac{\partial\phi_{\varphi}}{\partial\dot{a}}\left(\frac{\partial\alpha}{\partial a}\dot{a}+\frac{\partial\alpha}{\partial\varphi}\dot{\varphi}\right)-\frac{\partial\phi_{\varphi}}{\partial\dot{\varphi}}\left(\frac{\partial\beta}{\partial a}\dot{a}+\frac{\partial\beta}{\partial\varphi}\dot{\varphi}\right)=0.\nonumber 
\end{split}
\end{equation}
\end{theorem}

\begin{proof}
Computing the Lie derivative of $\varepsilon$ (\ref{eq:Source-EGC-Force})
with respect to $J^{2}\xi$, we obtain, in a~straightforward way, that
the Noether\textendash Bessel-Hagen equation is equivalent to the
following system of conditions,
\begin{equation}
\alpha+2a\frac{\partial\alpha}{\partial a}=0,\quad\frac{\partial\beta}{\partial a}a^{2}-6\frac{\partial\alpha}{\partial\varphi}=0,\quad3\alpha+2\frac{\partial\beta}{\partial\varphi}a=0.\label{eq:AlfaBeta}
\end{equation}

Equations (\ref{eq:AlfaBeta}) for components $\alpha$, $\beta$
of vector field $\xi$ can be integrated. From the first equation
of (\ref{eq:AlfaBeta}), we get
\[
\alpha=c\frac{1}{\sqrt{a}}\exp\left(h(\varphi)\right)
\]
for some function $h=h(\varphi)$ and some $c\in\mathbb{R}$. Differentiating
$\alpha$, we have
\begin{equation*}
\begin{split}
 & \frac{\partial\alpha}{\partial\varphi}=\alpha\frac{dh}{d\varphi},\quad\frac{\partial^{2}\alpha}{\partial\varphi^{2}}=\alpha\left(\frac{dh}{d\varphi}\frac{dh}{d\varphi}+\frac{d^{2}h}{d\varphi^{2}}\right),\\
 & \frac{\partial\alpha}{\partial a}=-\frac{1}{2a}\alpha,\quad\frac{\partial^{2}\alpha}{\partial a\partial\varphi}=-\frac{1}{2a}\alpha\frac{dh}{d\varphi},
\end{split}
\end{equation*}
and, from the second and third equations of (\ref{eq:AlfaBeta}), we
get
\[
\frac{\partial^{2}\beta}{\partial a\partial\varphi}=\frac{6}{a^{2}}\alpha\left(\frac{dh}{d\varphi}\frac{dh}{d\varphi}+\frac{d^{2}h}{d\varphi^{2}}\right)=\frac{9}{4a^{2}}\alpha=\frac{\partial^{2}\beta}{\partial\varphi\partial a},
\]
hence we get a second-order ordinary equation for $h=h(\varphi)$,
\[
\frac{d^{2}h}{d\varphi^{2}}+\frac{dh}{d\varphi}\frac{dh}{d\varphi}=\frac{3}{8},
\]
which is directly integrable with the solution
\begin{equation}
h(\varphi)=\ln\left(\exp\left(\sqrt{\frac{3}{2}}\varphi\right)+C\right)-\sqrt{\frac{3}{8}}\varphi+const\label{eq:h-funkce}
\end{equation}
for some real numbers $C>0$, $c_{0}\in\mathbb{R}$. Thus, $\alpha$
becomes of the form (\ref{eq:NBH-EGCsym1}), where $A=c\exp\left(c_{0}\right)$.

The third equation of (\ref{eq:AlfaBeta}) now reads
\[
\frac{\partial\beta}{\partial\varphi}=-\frac{3}{2a}A\frac{1}{\sqrt{a}}\left(\exp\left(\sqrt{\frac{3}{8}}\varphi\right)+C\exp\left(-\sqrt{\frac{3}{8}}\varphi\right)\right),
\]
which can be directly integrated, hence
\[
\beta=-\sqrt{6}A\frac{1}{a\sqrt{a}}\left(\exp\left(\sqrt{\frac{3}{8}}\varphi\right)-C\exp\left(-\sqrt{\frac{3}{8}}\varphi\right)\right)+p(a)
\]
for some function $p=p(a)$. Substituing this form of $\beta$ into
the second equation of (\ref{eq:AlfaBeta}), we easily observe that
$p$ must be constant, $p(a)=B\in\mathbb{R}$, that is, $\beta$ is
given by (\ref{eq:NBH-EGCsym2}). The pair $\alpha$, $\beta$ is
the unique solution of system (\ref{eq:AlfaBeta}). Note that components
of $\alpha$ (\ref{eq:NBH-EGCsym1}), $\beta$ (\ref{eq:NBH-EGCsym2})
of the~generator of invariance transformations of $\varepsilon$
are dependent,
\begin{equation*}
\begin{split}
 & \beta-B=-4\frac{1}{a}\frac{dh}{d\varphi}\alpha,
\end{split}
\end{equation*}
where $h=h(\varphi)$ is given by (\ref{eq:h-funkce}).
\end{proof}
In the next propositions, we analyze necessary and sufficient conditions
for potential function $V$ such that $\xi$ is a symmetry of Lagrangian
$\lambda$, and conditions for $V$ and external forces $(\phi_{a},\phi_{\varphi})$
such that vector field $\xi$ is a symmetry of source form $\varepsilon$.
Moreover, \emph{exact solutitons} of equations of motion (\ref{eq:ELeq1-EGC}),
(\ref{eq:ELeq2-EGC}), obeying energy condition (\ref{eq:EnergyEGC}),
are studied.
\begin{corollary}
\label{cor:SolLambda}Let $\lambda$ be the Klein-Gordon Lagrangian,
$\lambda=\mathscr{L}dt$ (\ref{eq:Lagrange-KG}). A~vector field
$\xi$ (\ref{eq:VectorField-2dim2}) on $TM$ is a~generator of invariance
transformations of $\lambda$ if and only if its coefficients $\alpha$,
$\beta$ are of the form (\ref{eq:NBH-EGCsym1}), (\ref{eq:NBH-EGCsym2}),
and one of the following conditions is satisfied:

\emph{(i)} $\alpha=0\,\,(A=0)$, $\beta=B\neq0$, and $V=V(\varphi)$
is constant. In this case, the equations of motion (\ref{eq:ELeq1-EGC}),
(\ref{eq:ELeq2-EGC}) and the energy condition (\ref{eq:EnergyEGC})
have a~solution for $V=0$ only, namely the mapping $\gamma:\mathbb{R}\rightarrow TM$
given by
\begin{equation}
\begin{split}
\gamma(t) & =\left(a(\gamma(t)),\varphi(\gamma(t))\right)=\left(k_{1}\sqrt[3]{3t+k_{2}},\pm\sqrt{\frac{2}{3}}\ln\left|3t+k_{2}\right|+k_{3}\right),\quad k_{1},k_{2},k_{3}\in\mathbb{R}.\label{eq:Sol1}
\end{split}
\end{equation}
The Noether conserved current (Corollary \ref{cor:Noether-Add}, (\ref{eq:ConservedLambda}))
along extremals (solutions of (\ref{eq:ELeq1-EGC}), (\ref{eq:ELeq2-EGC}))
reads
\[
L=g_{ij}\xi^{i}\dot{x}^{j}=-Ba^{3}\dot{\varphi}.
\]

\emph{(ii)} $\alpha\neq0\,\,(A\neq0),$ $B=0$, and $V=k\left(\exp\left(\sqrt{\frac{3}{8}}\varphi\right)-C\exp\left(-\sqrt{\frac{3}{8}}\varphi\right)\right)^{2}$,
where $k>0$. In this case, the equations of motion (\ref{eq:ELeq1-EGC}),
(\ref{eq:ELeq2-EGC}) and the energy condition (\ref{eq:EnergyEGC})
have a~solution $\gamma:\mathbb{R}\rightarrow TM$ given by 
\begin{equation*}
\begin{split}
 & \frac{\ddot{a}}{a}+2\frac{\dot{a}^{2}}{a^{2}}-V(\varphi)=0,\quad\ddot{\varphi}+3\frac{\dot{a}}{a}\dot{\varphi}+\frac{dV}{d\varphi}=0,\\
 & 3\frac{\dot{a}^{2}}{a^{2}}-\frac{1}{2}\dot{\varphi}^{2}-V(\varphi)=0.
\end{split}
\end{equation*}
The Noether conserved current (Corollary \ref{cor:Noether-Add}, (\ref{eq:ConservedLambda}))
along extremals (solutions of (\ref{eq:ELeq1-EGC}), (\ref{eq:ELeq2-EGC}))
reads
\[
L=g_{ij}\xi^{i}\dot{x}^{j}=6a\dot{a}\alpha.
\]
\end{corollary}

\begin{corollary}
\label{cor:NBHF}Let $\varepsilon$ (\ref{eq:Source-EGC-Force}) be
a source form on $\mathbb{R}\times T^{2}M$, associated to Klein-Gordon
Lagrangian $\lambda$ and an external force $\phi=(\phi_{a},\phi_{\varphi})$
on $TM$. A~vector field $\xi$ on $M$, expressed by (\ref{eq:VectorField-2dim2}),
is a~generator of invariance transformations of $\varepsilon$ if
and only if $\alpha$, $\beta$ are of the form (\ref{eq:NBH-EGCsym1}),
(\ref{eq:NBH-EGCsym2}), respectively, and $V=V(\varphi)$ and $(\phi_{a},\phi_{\varphi})$
satisfy the following

\emph{Case I.} $\phi_{a}=\phi_{\varphi}=0$:
\begin{eqnarray}
& \frac{3}{2}\alpha V+\beta a\frac{dV}{d\varphi}+2\frac{\partial\alpha}{\partial\varphi}\frac{dV}{d\varphi}=0,\label{eq:NBH-1-EGC1-1}\\
& \frac{3}{2}\alpha\frac{dV}{d\varphi}+\beta a\frac{d^{2}V}{d\varphi^{2}}+3\frac{\partial\alpha}{\partial\varphi}V=0.\label{eq:NBH-2-EGC1-1}
\end{eqnarray}

Non-trivial solutions of (\ref{eq:NBH-1-EGC1-1}), (\ref{eq:NBH-2-EGC1-1}),
read

\emph{(i)} $\alpha=0\,\,(A=0),\,\,B\neq0,\,\,V\,\,linear\,\,in\,\,\varphi$,

\emph{(ii)} $\alpha\neq0\,\,(A\neq0),\,\,B=0,\,\,V=k\left(\exp\left(\sqrt{\frac{3}{8}}\varphi\right)-C\exp\left(-\sqrt{\frac{3}{8}}\varphi\right)\right)^{2},\,\,k>0.$

Equations of motion (\ref{eq:ELeq1-EGC}), (\ref{eq:ELeq2-EGC}) and
the energy condition (\ref{eq:EnergyEGC}) have solutions described
by Proposition \ref{cor:SolLambda}. Moreover, if $V$ is linear,
i.e. $V=P\varphi+Q$, $P\neq0$, then (\ref{eq:ELeq1-EGC}), (\ref{eq:ELeq2-EGC}),
(\ref{eq:EnergyEGC}) possess a solution described by
\begin{equation*}
\begin{split}
 & \frac{\ddot{a}}{a}+2\frac{\dot{a}^{2}}{a^{2}}-P\varphi-Q=0,\quad\ddot{\varphi}+3\frac{\dot{a}}{a}\dot{\varphi}+P=0,\\
 & 3\frac{\dot{a}^{2}}{a^{2}}-\frac{1}{2}\dot{\varphi}^{2}-P\varphi-Q=0.
\end{split}
\end{equation*}

\emph{Case II.} $\phi_{a},\phi_{\varphi}$ are conservative, i.e.
$f_{a}=3a^{2}V(\varphi)$, $f_{\varphi}=a^{3}dV/d\varphi$:
\[
V=V(\varphi)\,\,is\,\,arbitrary.
\]
(\ref{eq:ELeq1-EGC}), (\ref{eq:ELeq2-EGC}), (\ref{eq:EnergyEGC})
possess a solution described by
\begin{equation*}
\begin{split}
 & 2\frac{\ddot{a}}{a}+\frac{\dot{a}^{2}}{a^{2}}+\frac{1}{2}\dot{\varphi}^{2}=0,\quad\ddot{\varphi}+3\frac{\dot{a}}{a}\dot{\varphi}=0,\\
 & V(\varphi)=3\frac{\dot{a}^{2}}{a^{2}}-\frac{1}{2}\dot{\varphi}^{2}.
\end{split}
\end{equation*}

\emph{Case III.} $\phi_{a},\phi_{\varphi}$ depend on $a$, $\varphi$
variables only:
\begin{eqnarray}
& \alpha\left(\frac{9}{2}aV-6a\frac{dh}{d\varphi}\frac{dV}{d\varphi}+\frac{1}{2a}f_{a}-\frac{\partial f_{a}}{\partial a}+4\frac{1}{a}\frac{dh}{d\varphi}\frac{\partial f_{a}}{\partial\varphi}-6\frac{1}{a^{2}}\frac{dh}{d\varphi}f_{\varphi}\right)\label{eq:NBH-1-EGC2}\\
& \quad+B\left(3a^{2}\frac{dV}{d\varphi}-\frac{\partial f_{a}}{\partial\varphi}\right)=0,\nonumber \\
& \alpha\left(3\frac{dh}{d\varphi}a^{2}V+\frac{3}{2}a^{2}\frac{dV}{d\varphi}-4\frac{dh}{d\varphi}a^{2}\frac{d^{2}V}{d\varphi^{2}}-\frac{dh}{d\varphi}f_{a}+\frac{3}{2a}f_{\varphi}-\frac{\partial f_{\varphi}}{\partial a}+4\frac{1}{a}\frac{dh}{d\varphi}\frac{\partial f_{\varphi}}{\partial\varphi}\right)\label{eq:NBH-2-EGC2}\\
& \quad+B\left(a^{3}\frac{d^{2}V}{d\varphi^{2}}-\frac{\partial f_{\varphi}}{\partial\varphi}\right)=0,\nonumber 
\end{eqnarray}
where $B\in\mathbb{R}$, and $h=h(\varphi)$ is function given by
(\ref{eq:h-funkce}).

A~non-trivial solution of (\ref{eq:NBH-1-EGC2}), (\ref{eq:NBH-2-EGC2}),
read
\begin{equation*}
\begin{split}
\alpha & =0\,\,(A=0),\quad\beta=B\in\mathbb{R},\quad f_{a}=3a^{2}V,\quad f_{\varphi}=a^{3}\frac{dV}{d\varphi},
\end{split}
\end{equation*}
where $V=V(\varphi)$ is arbitrary.

\emph{Case IV. }$\phi_{a},\phi_{\varphi}$ are dissipative, i.e. $\phi_{a}=-\phi_{11}\dot{a}-\phi_{12}\dot{\varphi}$,
$\phi_{\varphi}=-\phi_{12}\dot{a}-\phi_{22}\dot{\varphi}$ for some
functions $\phi_{ij}:M\rightarrow\mathbb{R}$:
\begin{equation*}
\begin{split}
 & \alpha\left(\frac{9}{2}aV-6a\frac{dh}{d\varphi}\frac{dV}{d\varphi}\right)+3Ba^{2}\frac{dV}{d\varphi}=0,\\
 & \alpha\left(3a^{2}\frac{dh}{d\varphi}V+\frac{3}{2}a^{2}\frac{dV}{d\varphi}-4a^{2}\frac{dh}{d\varphi}\frac{d^{2}V}{d\varphi^{2}}\right)+Ba^{3}\frac{d^{2}V}{d\varphi^{2}}=0,\\
 & \alpha\left(-\frac{1}{a}\phi_{11}+12\frac{1}{a^{2}}\frac{dh}{d\varphi}\phi_{12}+\frac{\partial\phi_{11}}{\partial a}-4\frac{1}{a}\frac{dh}{d\varphi}\frac{\partial\phi_{11}}{\partial\varphi}\right)+B\frac{\partial\phi_{11}}{\partial\varphi}=0,\\
 & \alpha\left(\frac{dh}{d\varphi}\phi_{11}-2\frac{1}{a}\phi_{12}+6\frac{1}{a^{2}}\frac{dh}{d\varphi}\phi_{22}+\frac{\partial\phi_{12}}{\partial a}-4\frac{1}{a}\frac{dh}{d\varphi}\frac{\partial\phi_{12}}{\partial\varphi}\right)+B\frac{\partial\phi_{12}}{\partial\varphi}=0,\\
 & \alpha\left(\frac{dh}{d\varphi}\phi_{11}-2\frac{1}{a}\phi_{12}+6\frac{1}{a^{2}}\frac{dh}{d\varphi}\phi_{22}+\frac{\partial\phi_{12}}{\partial a}-4\frac{1}{a}\frac{dh}{d\varphi}\frac{\partial\phi_{12}}{\partial\varphi}\right)+B\frac{\partial\phi_{12}}{\partial\varphi}=0,\\
 & \alpha\left(2\frac{dh}{d\varphi}\phi_{12}-3\frac{1}{a}\phi_{22}+\frac{\partial\phi_{22}}{\partial a}-4\frac{1}{a}\frac{dh}{d\varphi}\frac{\partial\phi_{22}}{\partial\varphi}\right)+B\frac{\partial\phi_{22}}{\partial\varphi}=0.
\end{split}
\end{equation*}

Non-trivial solutions of this system read

\emph{(i)} $\alpha=0\,\,(A=0),\,\,B\neq0$, $V$ linear in $\varphi$,
and force coefficients $\phi_{11},\phi_{12},\phi_{22}$ are functions
of variable $a$ only. 

\emph{(ii)} $\alpha\neq0\,\,(A\neq0),\,\,B=0$, $V=k\left(\exp\left(\sqrt{\frac{3}{8}}\varphi\right)-C\exp\left(-\sqrt{\frac{3}{8}}\varphi\right)\right)^{2},\,\,k>0$,
and $\phi_{11},\phi_{12},\phi_{22}$ satisfy

\begin{equation*}
\begin{split}
 & -\frac{1}{a}\phi_{11}+12\frac{1}{a^{2}}\frac{dh}{d\varphi}\phi_{12}+\frac{\partial\phi_{11}}{\partial a}-4\frac{1}{a}\frac{dh}{d\varphi}\frac{\partial\phi_{11}}{\partial\varphi}=0,\\
 & \frac{dh}{d\varphi}\phi_{11}-2\frac{1}{a}\phi_{12}+6\frac{1}{a^{2}}\frac{dh}{d\varphi}\phi_{22}+\frac{\partial\phi_{12}}{\partial a}-4\frac{1}{a}\frac{dh}{d\varphi}\frac{\partial\phi_{12}}{\partial\varphi}=0,\\
 & 2\frac{dh}{d\varphi}\phi_{12}-3\frac{1}{a}\phi_{22}+\frac{\partial\phi_{22}}{\partial a}-4\frac{1}{a}\frac{dh}{d\varphi}\frac{\partial\phi_{22}}{\partial\varphi}=0,
\end{split}
\end{equation*}
where $h=h(\varphi)$ is given by (\ref{eq:h-funkce}); this system
obey for instance $\phi_{11}=a$, $\phi_{12}=0$, $\phi_{22}=-\frac{1}{6}a^{3}$.
\end{corollary}
In conclusion, the existence of  Noether\textendash Bessel-Hagen symmetry selects the  the form of external force and self-interacting potentialof the scalar field.

\section{Discussion and Conclusions}
Starting from the  Noether\textendash Bessel-Hagen theorem, it is possible to find out symmetries for dynamical systems and then reducing them. The approach is particularly useful, with respect to the Noether standard,  for non-Lagrangian systems involving external forces or  dissipative terms. In fact, the Noether\textendash Bessel-Hagen symmetry can be directly searched for the equations of motion once the \eqref{eq:NBH-intro} condition is satisfied for the dynamical system.


In this work, we provide new analysis of the Noether--Bessel-Hagen equation for mechanical systems on $\mathbb{R}\times T^2M$ with external forces, and discuss necessary and sufficient conditions for these systems (namely with conservative, dissipative, or variational forces) to possess a~Noether--Bessel-Hagen symmetry. Conserved Noether currents for systems with variational external forces are also discussed. After developing the formal structure of the theory, we focused on two main examples: a damped harmonic oscillator and a  scalar-tensor  cosmological model. In the former case, we recover, as conserved quantity the angular momentum, as expected. 

In the latter example, we considered the field equations of a minimally coupled scalar-tensor model with unknown potential $V(\varphi)$ and external force terms. The approach leads to a system of differential equations which, once solved, according to the existence of a symmetry, provides the form of the scalar-field potential. We selected different shapes for $V(\varphi)$ showing that the space of solutions in this case, is different than the one obtained by applying the standard Noether Symmetry Approach to the cosmological Lagrangian. 

We presented a new way to look for symmetry of dynamical systems considering, specifically,   those systems which are not variational.  In this sense, the Noether\textendash Bessel-Hagen Symmetry Approach is a generalization of the Noether Symmetry Approach. In forthcoming papers,  we will apply the method to classes of cosmological models derived from  modified theories of gravity like those discussed in \cite{Cai:2015emx} and \cite{Capozziello:2011et}.

\begin{acknowledgements}ZU acknowledges support by the ESF in Science without borders project, reg. no. CZ.02.2.69/0.0./0.0./16\_027/0008463 within the Operational Programme Research, Development and Education. FB and SC acknowledge the support of Istituto Nazionale di Fisica Nucleare (INFN) (iniziative specifiche MOONLIGHT2, GINGER and QGSKY). This paper is based upon work from COST action CA15117 (CANTATA), supported by COST (European Cooperation in Science and Technology). 
\end{acknowledgements}

%

\section*{Data availability statement}
The data that support the findings of this study are available from the corresponding author upon reasonable request.

%


\begin{thebibliography}{99}


\bibitem{Capozziello:1996bi}
  S.~Capozziello, R.~De Ritis, C.~Rubano and P.~Scudellaro,
  Noether symmetries in cosmology,
  Riv.\ Nuovo Cim.\  {\bf 19N4} 1 (1996)

\bibitem{Dialektopoulos:2018qoe}
  K.~F.~Dialektopoulos and S.~Capozziello,
 Noether Symmetries as a geometric criterion to select theories of gravity,
  Int.\ J.\ Geom.\ Meth.\ Mod.\ Phys.\  {\bf 15}(No. supp. 01), 1840007 (2018) 




\bibitem{Bessel-Hagen}E. Bessel-Hagen, Uber die Erhaltungssatze der
Electrodynamik, Math. Ann. \textbf{84}, 258\textendash 276 (1921)

\bibitem{Kossmann}Y. Kossmann-Schwarzbach, \textcolor{black}{The
Noether Theorems}, Springer-Verlag, New York (2011)

\bibitem{Krupka-Book}D. Krupka, Introduction to Global Variational
Geometry, Atlantis Studies in Variational Geometry, Vol. 1, Atlantis
Press, Amsterdam\textendash Beijing\textendash Paris (2015)

\bibitem{BK1} J. Brajer\v{c}\'{i}k and D. Krupka, Variational principles for locally variational forms, J. Math. Phys. \textbf{46}, 052903 (2005)

\bibitem{BK2} J. Brajer\v{c}\'{i}k and D. Krupka, Cohomology and local variational principles, In: Proc. Conf. XV Internat.Workshop Geometry and Physics, Puerto de la Cruz, Canary Islands, Spain, September 11--16, 2006, Publ. de la RSME, 2007, 119--124.

\bibitem{UB} Z. Urban and J. Brajer\v{c}\'{i}k, The fundamental Lepage form in variational theory for submanifolds, Int. J. Geom. Meth. Mod. Phys. \textbf{15}, No. 6, 1850103 (2018)

\bibitem{Cattafi} F. Cattafi, M. Palese, and E. Winterroth, Variational derivatives in locally Lagrangian field theories and Noether-Bessel-Hagen currents, Int. J. Geom. Meth. Mod. Phys. \textbf{13}, No. 8, 1650067 (2016)

\bibitem{PW} M. Palese and E. Winterroth, Topological obstructions in Lagrangian field theories, with an application to 3D Chern-Simons gauge theory, J. Math. Phys. \textbf{58}, 023502 (2017)

\bibitem{Francaviglia} M. Francaviglia, M. Palese, and E. Winterroth, Variationally equivalent problems and variations of Noether currents, Int. J. Geom. Meth. Mod. Phys. \textbf{10}, No. 1, 1220024 (2013)

\bibitem{Bashirov} D. Bashkirov \textit{et al.}, Noether's second theorem in a general setting: reducible gauge theories, J. Phys. A: Math. Gen. \textbf{38} 5329 (2005)

\bibitem{Noether}E. Noether, Invariante Variationsprobleme, Nachr.
Konig. Gessell. Wissen. Gottingen, Math.-Phys. Kl. 235 (1918)


\bibitem{DK-Invariant}D. Krupka, Invariant variational structures
on fibered manifolds, Int. J. Geom. Meth. Mod. Phys. \textbf{12}, 1550020 (2015)

\bibitem{Trautman1}A. Trautman, Invariance of Lagrangian systems,
In: General Relativity, Papers in Honour of J.L. Synge, pp.  85\textendash 99. Oxford, Clarendon
Press (1972)

\bibitem{Trautman2}A. Trautman, Noether equations and conservation
laws, Commun. Math. Phys. \textbf{6}, 248\textendash 261 (1967)

\bibitem{Krupka-Invariance}D. Krupka, A Geometric Theory of Ordinary
First Order Variational Problems in Fibered Manifolds. II. Invariance,
J. Math. Anal. Appl. \textbf{49}, 469\textendash 476 (1975)

\bibitem{Sardanashvily}G. Sardanashvily, Noether's Theorems,
Applications in Mechanics and Field Theory, Atlantis Studies in Variational
Geometry, Vol. 3, Atlantis Press (2016)

\bibitem{Olver}P.J. Olver, Applications of Lie Groups to Differential
Equations, GTM 107, Springer-Verlag, New York (1986)

\bibitem{Bluman}G.W. Bluman and S. Kumei, Symmetries and Differential
Equations, Springer-Verlag, New York (1989)

\bibitem{Krupka-Forces}D. Krupka, Variational forces, Lepage Research
Institute Library \textbf{6}, 1\textendash 38 (2018)

\bibitem{Chien}N. Chien, T. Honein, and G. Herrmann, Dissipative
systems, conservation laws and symmetries, Int. J. Solids Structures
\textbf{33} (No. 20\textendash 22), 2959\textendash 2968 (1996) 

\bibitem{Honein}T. Honein, N. Chien, and G. Herrmann, On conservation
laws for dissipative systems, Physics Letter A. \textbf{155}, 223\textendash 224 (1991)

\bibitem{Capozziello:2007wc}
  S.~Capozziello, A.~Stabile and A.~Troisi,
  Spherically symmetric solutions in f(R)-gravity via Noether Symmetry Approach,
  Class.\ Quant.\ Grav.\  {\bf 24}, 2153 (2007) 

\bibitem{Paliathanasis:2014iva}
  A.~Paliathanasis, S.~Basilakos, E.~N.~Saridakis, S.~Capozziello, K.~Atazadeh, F.~Darabi and M.~Tsamparlis,
 New Schwarzschild-like solutions in f(T) gravity through Noether symmetries,
  Phys.\ Rev.\ D {\bf 89}, 104042 (2014)


\bibitem{Bajardi:2019zzs}
  F.~Bajardi, K.~F.~Dialektopoulos and S.~Capozziello,
  Higher dimensional static and spherically symmetric solutions in extended Gauss-Bonnet gravity,
  arXiv:1911.03554 [gr-qc].


\bibitem{Capozziello:2012iea}
  S.~Capozziello, N.~Frusciante and D.~Vernieri,
  New Spherically Symmetric Solutions in f(R)-gravity by Noether Symmetries,
  Gen.\ Rel.\ Grav.\  {\bf 44}, 1881 (2012) 


\bibitem{Capozziello:2008ch}
  S.~Capozziello and A.~De Felice,
  f(R) cosmology by Noether's symmetry,
  JCAP {\bf 0808}, 016 (2008) 


\bibitem{Capozziello:1996ay}
  S.~Capozziello, G.~Marmo, C.~Rubano and P.~Scudellaro,
  Noether symmetries in Bianchi universes,
  Int.\ J.\ Mod.\ Phys.\ D {\bf 6}, 491 (1997) 
  [gr-qc/9606050].

\bibitem{Capozziello}S. Capozziello, M. De Laurentis, and S.D. Odintsov,
Hamiltonian dynamics and Noether symmetries in Extended Gravity Cosmology,
The European Physical Journal C \textbf{72}, 2068 (2012)

\bibitem{Atazadeh:2011aa}
  K.~Atazadeh and F.~Darabi,
  $f(T)$ cosmology via Noether symmetry,
  Eur.\ Phys.\ J.\ C {\bf 72}, 2016 (2012) 
  [arXiv:1112.2824 [physics.gen-ph]].

\bibitem{KUV}D. Krupka, Z. Urban, and J. Voln\'a, Variational submanifolds
of Euclidean spaces, J. Math. Phys. \textbf{59}(3), 032903 (2018)


\bibitem{UV}Z. Urban and J. Voln\'{a}, On a global Lagrangian construction
for ordinary variational equations on 2-manifolds, J. Math.
Phys. \textbf{60}(9) 092902 (2019)

\bibitem{Volna}J. Voln\'{a} and Z. Urban, First-order Variational Sequences
in Field Theory, in: D. Zenkov (Ed.), \textcolor{black}{The
Inverse Problem of the Calculus of Variations}, \textcolor{black}{Local
and Global Theory}, pp. 215\textendash 284, Atlantis Press, Amsterdam\textendash Beijing\textendash Paris (2015)

\bibitem{KKS}D. Krupka, O. Krupkov\'a, and D. Saunders, Cartan\textendash Lepage
forms in geometric mechanics, Int. J. Non-Linear Mech. \textbf{47}, 1154\textendash 1160 (2012)

\bibitem{Landau}L.D. Landau and E.M. Lifshitz, Mechanics,
Course of Theoretical Physics, Vol. 1, Pergamon Press, Oxford (1969)

\bibitem{Green:1987sp}
  M.~B.~Green, J.~H.~Schwarz and E.~Witten,
  Superstring Theory. Vol. 1: Introduction,
  Cambridge, Uk: Univ. Pr. ( 1987) 469 P. ( Cambridge Monographs On Mathematical Physics)
  
\bibitem{Green:1987mn}
  M.~B.~Green, J.~H.~Schwarz and E.~Witten,
  Superstring Theory. Vol. 2: Loop Amplitudes, Anomalies And Phenomenology,
  Cambridge, Uk: Univ. Pr. ( 1987) 596 P. (Cambridge Monographs On Mathematical Physics)
  
\bibitem{Polchinski:1998rq}
  J.~Polchinski,
  String theory. Vol. 1: An introduction to the bosonic string, Cambridge University Press (1998)
  
\bibitem{Polchinski:1998rr}
  J.~Polchinski,
  String theory. Vol. 2: Superstring theory and beyond, Cambridge University Press (1998)
  
\bibitem{Becker:2007zj}
  K.~Becker, M.~Becker and J.~H.~Schwarz,
  String theory and M-theory: A modern introduction, Cambridge University Press (2006)
  
\bibitem{Clifton:2011jh}
  T.~Clifton, P.~G.~Ferreira, A.~Padilla and C.~Skordis,
 Modified Gravity and Cosmology,
  Phys.\ Rept.\  {\bf 513}, 1 (2012)
  
\bibitem{Han:1998sg}
  T.~Han, J.~D.~Lykken and R.~J.~Zhang,
  On Kaluza-Klein states from large extra dimensions,
  Phys.\ Rev.\ D {\bf 59}, 105006 (1999)
    
\bibitem{Rovelli:2014ssa}
  C.~Rovelli and F.~Vidotto,
Covariant Loop Quantum Gravity : An Elementary Introduction to Quantum Gravity and Spinfoam Theory, Cambridge University Press (2014)
  
\bibitem{Ashtekar:2011ni}
  A.~Ashtekar and P.~Singh,
Loop Quantum Cosmology: A Status Report,
  Class.\ Quant.\ Grav.\  {\bf 28}, 213001 (2011) 
  
\bibitem{Rovelli:1997yv}
  C.~Rovelli,
  Loop quantum gravity,
  Living Rev.\ Rel.\  {\bf 1}, 1 (1998)
  

\bibitem{Kiritsis:2009sh}
  E.~Kiritsis and G.~Kofinas,
  Horava-Lifshitz Cosmology,
  Nucl.\ Phys.\ B {\bf 821}, 467 (2009) 
  
\bibitem{Cai:2009pe}
  R.~G.~Cai, L.~M.~Cao and N.~Ohta,
Topological Black Holes in Horava-Lifshitz Gravity,
  Phys.\ Rev.\ D {\bf 80}, 024003 (2009) 
  
\bibitem{Sotiriou:2010wn}
  T.~P.~Sotiriou,
  Horava-Lifshitz gravity: a status report,
  J.\ Phys.\ Conf.\ Ser.\  {\bf 283}, 012034 (2011) 
  
\bibitem{Mukohyama:2010xz}
  S.~Mukohyama,
  Horava-Lifshitz Cosmology: A Review,
  Class.\ Quant.\ Grav.\  {\bf 27}, 223101 (2010) 
  
\bibitem{ArkaniHamed:2002fu}
  N.~Arkani-Hamed, S.~Dimopoulos, G.~Dvali and G.~Gabadadze,
Nonlocal modification of gravity and the cosmological constant problem, arXiv: hep-th/0209227 (2002)
  
\bibitem{Modesto:2013jea}
  L.~Modesto and S.~Tsujikawa,
  Non-local massive gravity,
  Phys.\ Lett.\ B {\bf 727}, 48 (2013) 
  
\bibitem{Modesto:2013ioa}
  L.~Modesto,
  Super-renormalizable Gravity,''
$doi:10.1142/9789814623995 \_ 0098$ (2013)
  
\bibitem{Sotiriou:2008rp}
  T.~P.~Sotiriou and V.~Faraoni,
  f(R) Theories Of Gravity,
  Rev.\ Mod.\ Phys.\  {\bf 82}, 451 (2010) 

\bibitem{DeFelice:2010aj}
  A.~De Felice and S.~Tsujikawa,
  f(R) theories,
  Living Rev.\ Rel.\  {\bf 13}, 3 (2010) 


\bibitem{Cai:2015emx}
  Y.~F.~Cai, S.~Capozziello, M.~De Laurentis and E.~N.~Saridakis,
  f(T) teleparallel gravity and cosmology,
  Rept.\ Prog.\ Phys.\  {\bf 79}(10), 106901 (2016)

\bibitem{Hammond:2002rm}
  R.~T.~Hammond,
 Torsion gravity,
  Rept.\ Prog.\ Phys.\  {\bf 65}, 599 (2002) 

\bibitem{Arcos:2005ec}
  H.~I.~Arcos and J.~G.~Pereira,
  Torsion gravity: A Reappraisal,
  Int.\ J.\ Mod.\ Phys.\ D {\bf 13}, 2193 (2004) 
  
  
  
\bibitem{Capozziello:2019klx}
  S.~Capozziello and F.~Bajardi,
  Gravitational waves in modified gravity,
  Int.\ J.\ Mod.\ Phys.\ D {\bf 28}(5), 1942002 (2019)


\bibitem{Capozziello:2011et}
  S.~Capozziello and M.~De Laurentis,
  Extended Theories of Gravity,
  Phys.\ Rept.\  {\bf 509}, 167 (2011) 


\bibitem{Capozziello:2007ec}
  S.~Capozziello and M.~Francaviglia,
  Extended Theories of Gravity and their Cosmological and Astrophysical Applications,
  Gen.\ Rel.\ Grav.\  {\bf 40}, 357 (2008) 

\bibitem{Clifton:2011jh}
  T.~Clifton, P.~G.~Ferreira, A.~Padilla and C.~Skordis,
  Modified Gravity and Cosmology,
  Phys.\ Rept.\  {\bf 513}, 1 (2012) 
  
  
\bibitem{Guth:1980zm}
  A.~H.~Guth,
 The Inflationary Universe: A Possible Solution to the Horizon and Flatness Problems,
  Phys.\ Rev.\ D {\bf 23}, 347 (1981) 
   [Adv.\ Ser.\ Astrophys.\ Cosmol.\  {\bf 3}, 139 (1987) ].

\bibitem{Guth:1982ec}
  A.~H.~Guth and S.~Y.~Pi,
  Fluctuations in the New Inflationary Universe,
  Phys.\ Rev.\ Lett.\  {\bf 49}, 1110 (1982) 

\bibitem{Linde:1981mu}
  A.~D.~Linde,
  A New Inflationary Universe Scenario: A Possible Solution of the Horizon, Flatness, Homogeneity, Isotropy and Primordial Monopole Problems,
  Phys.\ Lett.\  {\bf 108B}, 389 (1982) 
   [Adv.\ Ser.\ Astrophys.\ Cosmol.\  {\bf 3}, 149 (1987)].
  
\bibitem{Copeland:2006wr}
  E.~J.~Copeland, M.~Sami and S.~Tsujikawa,
  Dynamics of dark energy,
  Int.\ J.\ Mod.\ Phys.\ D {\bf 15}, 1753 (2006) 


\bibitem{Bezrukov:2007ep}
  F.~L.~Bezrukov and M.~Shaposhnikov,
 The Standard Model Higgs boson as the inflaton,
  Phys.\ Lett.\ B {\bf 659}, 703 (2008) 
  

\bibitem{Capozziello:1993vr}
  S.~Capozziello and R.~de Ritis,
Relation between the potential and nonminimal coupling in inflationary cosmology,
  Phys.\ Lett.\ A {\bf 177}, 1 (1993)


\bibitem{Capozziello:1994du}
  S.~Capozziello and R.~de Ritis,
Noether's symmetries and exact solutions in flat nonminimally coupled cosmological models,
  Class.\ Quant.\ Grav.\  {\bf 11}, 107 (1994) 


\bibitem{Paliathanasis:2014rja}
  A.~Paliathanasis, M.~Tsamparlis, S.~Basilakos and S.~Capozziello,
 Scalar-Tensor Gravity Cosmology: Noether symmetries and analytical solutions,
  Phys.\ Rev.\ D {\bf 89}(6), 063532 (2014)


\bibitem{Capozziello:1993yy}
  S.~Capozziello, R.~de Ritis and P.~Scudellaro,
 Noether's symmetries in (n+1)-dimensional nonminimally coupled cosmologies,
  Int.\ J.\ Mod.\ Phys.\ D {\bf 2}, 463 (1993) 


\bibitem{Borowiec:2014wva}
  A.~Borowiec, S.~Capozziello, M.~De Laurentis, F.~S.~N.~Lobo, A.~Paliathanasis, M.~Paolella and A.~Wojnar,
 Invariant solutions and Noether symmetries in Hybrid Gravity,
  Phys.\ Rev.\ D {\bf 91}(2), 023517 (2015)
  
  \bibitem{ester}
  E.~Piedipalumbo, M.~De Laurentis and S.~Capozziello,
 Noether symmetries in Interacting Quintessence Cosmology,
  Phys.\ Dark Univ.\  {\bf 27}, 100444 (2020) 
  

  
  
%
%
\end{thebibliography}


\end{document}